\documentclass[letterpaper,english,twocolumn]{IEEEtran}
\usepackage[T1]{fontenc}
\usepackage{verbatim}
\usepackage{float}
\usepackage{bm}
\usepackage{amsmath}
\usepackage{amsthm}
\usepackage{amssymb}

\newif\ifextpdf
\extpdftrue

\makeatletter

\pdfpageheight\paperheight
\pdfpagewidth\paperwidth

\providecommand{\tabularnewline}{\\}
\floatstyle{ruled}
\newfloat{algorithm}{tbp}{loa}
\providecommand{\algorithmname}{Algorithm}
\floatname{algorithm}{\protect\algorithmname}

\theoremstyle{plain}
\newtheorem{thm}{\protect\theoremname}
\theoremstyle{definition}
\newtheorem{defn}[thm]{\protect\definitionname}
\theoremstyle{remark}
\newtheorem{rem}[thm]{\protect\remarkname}
\theoremstyle{plain}
\newtheorem{lem}[thm]{\protect\lemmaname}
\theoremstyle{plain}
\newtheorem{cor}[thm]{\protect\corollaryname}
\theoremstyle{plain}
\newtheorem{fact}[thm]{\protect\factname}

\usepackage{hyperref}

\IEEEoverridecommandlockouts
\usepackage{amsthm}

\usepackage{enumitem}

\usepackage{tikz}
\usepackage{tkz-berge}
\usepackage{pgfplots}

\usetikzlibrary{external}
\ifextpdf
\else
\fi

\pgfplotsset{compat=newest} \pgfplotsset{plot coordinates/math
parser=false}

\usetikzlibrary{arrows,decorations.pathmorphing,backgrounds,positioning,fit}

\usetikzlibrary{plotmarks,shapes.geometric}

\usetikzlibrary{arrows,decorations.pathmorphing,decorations.footprints,fadings,calc,trees,mindmap,shadows,decorations.text,patterns,positioning,shapes,matrix,fit}
\graphicspath{{./}}

\setlist[description]{labelindent=0.25in,labelwidth=0.6in}

\makeatother

\usepackage{babel}
\providecommand{\corollaryname}{Corollary}
\providecommand{\definitionname}{Definition}
\providecommand{\factname}{Fact}
\providecommand{\lemmaname}{Lemma}
\providecommand{\remarkname}{Remark}
\providecommand{\theoremname}{Theorem}

\begin{document}

\title{Approaching Capacity at High Rates\\
with Iterative Hard-Decision Decoding}

\author{Yung-Yih Jian, Henry D. Pfister, and Krishna R. Narayanan\thanks{This work was done while the authors were with the Department of Electrical and Computer Engineering at Texas A\&M University. The material is based upon work supported by the National Science Foundation (NSF) under Grants No. 0747470 and No. 1320924. Any opinions, findings, conclusions, and recommendations expressed in this material are those of the authors and do not necessarily reflect the views of these sponsors. This research was presented in part at the 2012 International Symposium on Information Theory in Cambridge, Massachusetts~\cite{Jian-isit12} and the 2013 IEEE GLOBECOM conference in Atlanta, Georgia~\cite{Jian-globe13}.

Yung-Yih Jian is currently with Qualcomm Inc. in Santa Clara, CA (e-mail: yungyih.jian@gmail.com). Henry D. Pfister is currently an associate professor in the Department of Electrical and Computer Engineering at Duke University (email: henry.pfister@duke.edu). Krishna R. Narayanan is currently a full professor in the Department of Electrical and Computer Engineering at Texas A\&M University (email: krn@tamu.edu).}
}

\maketitle
\global\long\def\emap{\epsilon^{\mathrm{MAP}}}

\global\long\def\ebp{\epsilon^{\mathrm{BP}}}

\global\long\def\pbmap{p^{\mathrm{BMAP}}}

\global\long\def\pmp{p^{\mathrm{MP}}}

\global\long\def\phd{p^{\mathrm{HD}}}

\global\long\def\hmp{h^{\mathrm{MP}}}

\global\long\def\hbmap{h^{\mathrm{BMAP}}}

\global\long\def\xmap{x^{\mathrm{MAP}}}

\global\long\def\xbp{x^{\mathrm{BP}}}

\global\long\def\lbmap{\lambda^{\mathrm{BMAP}}}

\global\long\def\xmp{x^{\mathrm{MP}}}

\global\long\def\d{\mathrm{d}}
\global\long\def\hx{\hat{x}}
\global\long\def\hf{\hat{f}}
\global\long\def\ovU{\overline{U}}

\global\long\def\ovp{\overline{p}}

\global\long\def\itLambda{\mathit{\Lambda}}
\global\long\def\itTheta{\mathit{\Theta}}
\global\long\def\Cc{\mathcal{C}}

\global\long\def\hfn{\hat{f}_{n}}
\global\long\def\hpn{\hat{p}_{n}}

\global\long\def\hl{\hat{\lambda}}
\global\long\def\hr{\hat{\rho}}
\global\long\def\Uh{\hat{U}}
\global\long\def\ovl{\overline{\lambda}}
\global\long\def\ovr{\overline{\rho}}

\global\long\def\Dec{\mathsf{D}}
\global\long\def\nbr{\mathcal{N}}
\global\long\def\iDec{\hat{\mathsf{D}}}
\global\long\def\dmin{d_{\mathrm{min}}}

\global\long\def\A{\boldsymbol{A}}
\global\long\def\x{\boldsymbol{x}}
\global\long\def\y{\boldsymbol{y}}

\global\long\def\tf{\tilde{f}}
\global\long\def\tg{\tilde{g}}

\global\long\def\hVn{\hat{U_{n}}}
\global\long\def\hV{\hat{U}}
\global\long\def\hU{\hat{V}}

\global\long\def\vect#1{\bm{#1}}

\global\long\def\set#1{\mathcal{#1}}

\global\long\def\op#1{\mathtt{#1}}

\ifextpdf
\else
\input{./graphical_settings}
\fi

\begin{abstract}
A variety of low-density parity-check (LDPC) ensembles have now been observed to approach capacity with message-passing decoding. However, all of them use soft (i.e., non-binary) messages and a posteriori probability (APP) decoding of their component codes. In this paper, we show that one can approach capacity at high rates using iterative hard-decision decoding (HDD) of generalized product codes. Specifically, a class of spatially-coupled GLDPC codes with BCH component codes is considered, and it is observed that, in the high-rate regime, they can approach capacity under the proposed iterative HDD. These codes can be seen as generalized product codes and are closely related to braided block codes. An iterative HDD algorithm is proposed that enables one to analyze the performance of these codes via density evolution (DE).
\end{abstract}


\section{Introduction}

In his groundbreaking 1948 paper, Shannon defined the capacity of a noisy channel as the largest information rate for which reliable communication is possible \cite{Shannon-bell48}. Since then, researchers have spent countless hours looking for ways to achieve this rate in practical systems. In the 1990s, the problem was essentially solved by the introduction of iterative soft decoding for turbo and low-density parity-check (LDPC) codes \cite{Berrou-icc93,Luby-it01,Richardson-it01*2}. Although the decoding complexity is significant, these new codes were adopted quickly in wireless communication systems where the data rates were not too large \cite{Bender-commag00,Douillard-brest00}. In contrast, complexity issues have slowed their adoption in very high-speed systems, such as those used in optical and wireline communication. 

Introduced by Gallager in 1960, LDPC codes are linear block codes defined by a sparse parity-check matrix \cite{Gallager-60}. Using the parity-check matrix, an $(N,K)$ LDPC code can be represented by a Tanner graph, which is a bipartite graph with $N$ bit nodes and $N-K$ check nodes. The check nodes in the Tanner graph of an LDPC code represent the constraint that the group of bit nodes connected to a check node should form a codeword in a single-parity check (SPC) code. In 1981, Tanner generalized LDPC codes by replacing the SPC constraint nodes with more general constraints \cite{Tanner-it81}. Particularly, the bit nodes connected to a check node are constrained to be codewords of $(n,k)$ linear block codes such as Hamming codes, Bose-Chaudhuri-Hocquengham (BCH) codes or Reed-Solomon codes. After their introduction by Tanner, generalized LDPC (GLDPC) codes remained largely unexplored until the work of Boutros \emph{et al.} \cite{Boutros-icc99} and Lentmaier and Zigangirov \cite{Lentmaier-comlett99}. 

GLDPC codes can have both large minimum distance and good iterative decoding thresholds \cite{Miladinovic-com08}. But, the per-iteration decoding complexity of belief-propagation (BP) decoding of GLDPC codes is typically much higher than LDPC codes since optimal soft-input soft-output (SISO) decoding has to be performed for each component block code. However, the number of iterations required for the decoding algorithm to converge can be substantially smaller. Recently, generalized product codes, \emph{i.e.,} GLDPC codes with degree-$2$ bits, have been widely considered in optical communication systems \cite{itug9751}. In \cite{Djordjevic-jlt05}, GLDPC codes were proposed for 40Gb/s optical transport networks and it was shown that these codes outperform turbo codes by about 1 dB at a rate of 0.80. As such, GLDPC codes can provide high coding gains. But, if the full BP decoder is used, then the decoding complexity is still prohibitively high for implementation in very high-speed systems. 

In this paper, we show that, by using iterative \emph{hard-decision decoding} (HDD) of generalized product codes with BCH component codes, one can approach the capacity of the binary symmetric channel (BSC) in the high-rate regime. We consider an ensemble of spatially-coupled GLDPC codes based on $t$-error correcting BCH codes. For the BSC, we show that the redundancy-threshold tradeoff of this ensemble, under iterative HDD, scales optimally in the high-rate regime. To the best of our knowledge, this is the first example of an iterative HDD system that can provably approach capacity.  It is interesting to note that iterative HDD of product codes was first proposed well before the recent revolution in iterative decoding but the performance gains were limited~\cite{Abramson-comtech68}. Iterative decoding of product codes became competitive only after the advent of iterative soft decoding based on the turbo principle~\cite{Lodge-icc93,Pyndiah-com98}. A modified iterative HDD for GLDPC codes was also proposed by Miladinovic and Fossorier in~\cite{Miladinovic-com08} and improved threshold performance was observed. 

Under the assumption that the component code decoder corrects all patterns of $t$ or fewer errors and leaves all other cases unchanged, the asymptotic noise threshold for product codes has been studied in~\cite{Schwartz-isit05,Justesen-itw07}. In~\cite{Schwartz-isit05}, Schwartz \emph{et al.} analyze the asymptotic block error probability for product codes using combinatorial arguments. By using random graph arguments, another asymptotic threshold analysis, based on the result of the existence of ``$k$-core'' in a random graph~\cite{Pittel-jctb96}, is proposed by Justesen \emph{et al.}~\cite{Justesen-itw07}. Finally, counting arguments are used in~\cite{Barg-it11} to analyze the iterative HDD of GLDPC codes for adversarial error patterns and, hence, somewhat lower thresholds are reported.

Convolutional LDPC codes (or spatially-coupled LDPC codes) were introduced in~\cite{Felstrom-it99} and later discovered to achieve the MAP threshold under iterative decoding~\cite{Lentmaier-it10,Kudekar-it11}. Spatially-coupled GLDPC codes were introduced and analyzed in~\cite{Lentmaier-isit10}. Our choice of ensemble was motivated by the generalized product codes now used in optical communications \cite{itug9751} and their similarity to braided block codes \cite{Truhachev-isit03,Feltstrom-it09}. In particular, we consider the iterative HDD of spatially-coupled generalized product codes with $t$-error correcting component codes. This is similar to other recent work on coding system for optical communication systems \cite{Justesen-commag10,Justesen-toc11,Smith-jlt12,Jian-globe13}. The main difference is that the proposed iterative HDD updates messages using only the extrinsic information. Therefore, HDD of our spatially-coupled GLDPC ensemble can be rigorously analyzed via density evolution (DE) even when miscorrection occurs. This type of analysis actually dates back to Tanner, who applied it to uncoupled GLDPC codes with Hamming component codes in~\cite{Tanner-it81}. This DE analysis also allows us to show that iterative HDD can approach capacity in the high-rate regime. Also, for generalized product codes, a practical implementation of the proposed iterative HDD is introduced.

It is worth noting that a number of recent papers consider interesting variations  of this work~\cite{Zhang-isit15,Truhachev-isit16,Hager-isit16,Hager-it17}.

\vspace{0mm}

\section{Ensembles and Decoding Algorithms\label{sec: Ensembles}}

In this section, various code ensembles and decoding algorithms are introduced. We first recall the GLDPC ensemble. Based on the GLDPC ensemble, the spatially-coupled GLDPC ensemble is introduced. Also, a modified iterative HDD algorithm for GLDPC codes is proposed in this section. Since the proposed iterative HDD updates hard-decision messages only from extrinsic hard-decision messages, the performance of the proposed iterative HDD can be analyzed by DE. An ideal iterative HDD algorithm is also discussed, and its DE is described for the purpose of comparing with the proposed iterative HDD.

\subsection{Ensembles}

Let $\mathcal{C}$ be an $(n,k,\dmin)$ binary linear code that can correct all error patterns of weight at most $t$ (i.e., $\dmin\geq2t+1$). For example, one might choose $\mathcal{C}$ to be a primitive BCH code with parameters $(2^{\nu}-1,2^{\nu}-\nu t-1,2t+1)$. Now, we consider a GLDPC ensemble where every bit node satisfies two code constraints defined by $\mathcal{C}$. 
\begin{defn}
Each element of the $\left(\mathcal{C},m\right)$ GLDPC ensemble is defined by a Tanner graph shown in Figure \ref{fig: GLDPC} and denoted by $\set G=\left(\set I\cup\set J,\set E\right)$. There are $N=\frac{mn}{2}$ degree-2 bit nodes in set $\set I$, and $m$ degree-$n$ code-constraint (or constraint) nodes defined by $\mathcal{C}$ in set $\set J$. A random element from the ensemble is constructed by using an uniform random permutation for the $mn$ edges from the bit nodes to the constraint nodes. From the construction of the code, one can show that the design rate of $(\Cc,m)$ ensemble is 
\[
R=\frac{N-m(n-k)}{N}=1-\frac{2(n-k)}{n}=2\frac{k}{n}-1.
\]

\begin{figure}[t]
\begin{center}
\ifextpdf
\includegraphics{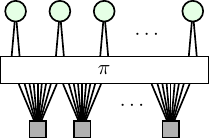}
\else
\input{./tikz_GLDPC}
\fi
\end{center}
\vspace{-4mm}\caption{A $(\protect\Cc,m)$ GLDPC ensemble, where $\pi$ is a random permutation.\label{fig: GLDPC}}
\end{figure}
\end{defn}
Now, we consider a spatially-coupled GLDPC ensemble where every bit node satisfies two code constraints defined by $\mathcal{C}$. Similar to the definition introduced in \cite{Kudekar-it11}, the spatially-coupled GLDPC ensemble $\left(\mathcal{C},m,L,w\right)$ is defined as follows.
\begin{defn}
\label{def: SCEnsemble} %
{} The Tanner graph of an element of the $\left(\mathcal{C},m,L,w\right)$ spatially-coupled GLDPC ensemble contains $L$ positions, $\{1,2,\dots,L\}$, of bit nodes and $L+w-1$ positions, $\{1,2,\dots,L+w-1\}$, of code-constraint nodes defined by $\set C$. Let $m$ be chosen such that $mn$ is divisible by both $2$ and $w$. At each position, there are $N=\frac{mn}{2}$ degree-$2$ bit nodes and $m$ degree-$n$ code-constraint nodes. A random element of the $\left(\mathcal{C},m,L,w\right)$ spatially-coupled GLDPC ensemble is constructed %
as follows. At each bit position and code-constraint position, the $mn$ sockets are partitioned into $w$ groups of $\frac{mn}{w}$ sockets via a uniform random permutation. Let $\set S_{i,j}^{(b)}$ and $\set S_{i,j}^{(c)}$ be, respectively, the $j$-th group at the $i$-th bit position and the $j$-th group at $i$-th code-constraint position, where $j\in\{0,1,\dots,w-1\}$. The Tanner graph is constructed by connecting $\set S_{i,j}^{(b)}$ to $\set S_{i+j,w-j-1}^{(c)}$ (\emph{i.e.}, by mapping the $\frac{mn}{w}$ edges between the two groups). An example of the $\left(\mathcal{C},m,L,w\right)$ spatially-coupled GLDPC ensemble with $w=3$ is shown in Figure \ref{fig: SC-GLDPC}.

\begin{figure*}[t]
\begin{center}
\ifextpdf
\includegraphics{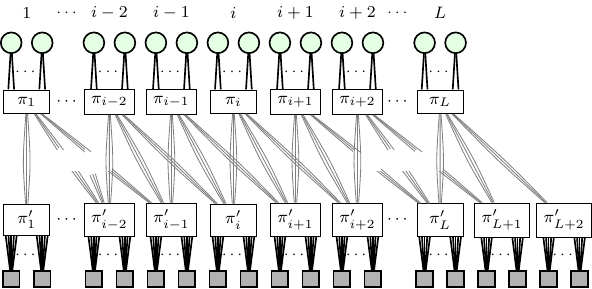}
\else
\input{./tikz_SC-GLDPC}
\fi
\end{center}
\vspace{-4mm}\caption{An example of $\left(\mathcal{C},m,L,w\right)$ spatially-coupled GLDPC ensemble, where $w=3$, and $(\pi_{i},\pi'_{i})$ are random permutations at position $i$ for bit nodes and constraint nodes, respectively. \label{fig: SC-GLDPC}}
\end{figure*}
\end{defn}
\begin{rem}
Since extra constraint nodes are required for spatial coupling, the design rate of the spatially-coupled ensemble is smaller than the design rate of the underlying ensemble \cite{Kudekar-it11}. According to the construction in Definition \ref{def: SCEnsemble}, $m(w-1)$ new constraint nodes are added after coupling. Thus, there are $NL$ bit nodes and $m(L+w-1)$ constraint nodes in the spatially-coupled code. The resulting design rate is at least 
\begin{align}
R_{SC} & \geq\frac{NL-m(L+w-1)(n-k)}{NL}\nonumber \\
 & =1-\frac{2(n-k)}{n}\left(1+\frac{w-1}{L}\right)\label{eq: SCR}\\
 & =R-(1-R)\frac{w-1}{L},\label{eq: SCR2}
\end{align}
where the second term in (\ref{eq: SCR2}) is the rate loss due to adding constraint nodes. One can see that the rate loss vanishes as $L\to\infty$. We note that the actual rate, which is defined as the ratio of the dimension of the code and the codeword length, may be slightly higher due to the implied shortening of the code constraints.
\end{rem}
\vspace{0mm}

\subsection{Iterative HDD with Ideal Component Decoders\label{subsec: IdealHDD}}

In this section, an iterative HDD with an ideal (i.e., genie aided) component-code decoder is introduced. This decoder corrects bits only when the number of error bits is less than or equal to the decoding radius $t$. In particular, the aid of a genie allows the ideal decoder to avoid miscorrection. To be explicit, we define $\iDec\colon\Cc\times\{0,1\}^{n}\rightarrow\{0,1\}^{n}$ as the operator of the ideal decoder. Given that a codeword $\vect c\in\Cc$ is transmitted, let $\vect e\in\{0,1\}^{n}$ be a binary error vector and $\vect v\triangleq\vect c\oplus\vect e$ be the received vector. Then, the output of the ideal decoder is 
\begin{align*}
\iDec(\vect c,\vect e) & \triangleq\begin{cases}
\vect c & \mbox{ if }d_{H}(\vect 0,\vect e)\leq t\\
\vect c\oplus\vect e & \mbox{ otherwise,}
\end{cases}
\end{align*}
where $d_{H}(\cdot,\cdot)$ is the Hamming distance between the two arguments. Also, the bit-level mapping implied by the ideal decoder, denoted by $\iDec_{i}\colon\Cc\times\{0,1\}^{n}\rightarrow\{0,1\}$, maps $(\vect c,\vect e)$ to the $i$-th bit of $\iDec(\vect c,\vect e)$. The decoder performance is independent of the transmitted codeword since the component decoder satisfies the symmetry condition, \emph{i.e.,} $\iDec(\vect c\oplus\vect c',\vect e)=\iDec(\vect c,\vect e)\oplus\vect c'$ for all $\vect c'\in\Cc$. In this case, the iterative HDD can be analyzed (e.g., using DE) under the assumption that the all-zero codeword is transmitted. The DE analysis for iterative HDD with ideal component decoders is discussed in Section \ref{subsec: IdealDE}. Under the all-zero codeword assumption, we can also define a simplified ideal decoding function $\iDec(\vect e)\triangleq\iDec(\vect 0,\vect e)$ that takes only one argument.

\begin{figure}[t]
\begin{center}
\ifextpdf
\includegraphics{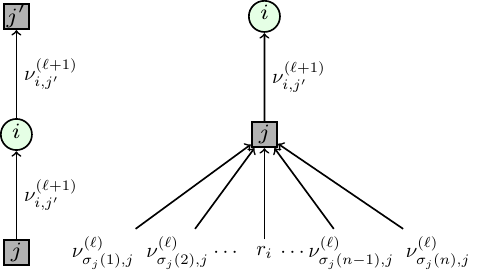}
\else
\input{./tikz_IterHDD}
\fi
\end{center}
\vspace{-4mm}

\caption{The proposed iterative HDD. At the $i$-th variable node, the input message from the $j$-th constraint node is forwarded to the $j'$-th constraint node. At the $j$-th constraint node, the input corresponding to the $i$-th variable node is replaced by $r_{i}$. \label{fig:IterHDD}}
\end{figure}

Now, we start to describe the ideal iterative HDD under the assumption that the all-zero codeword was transmitted. Decoding proceeds by passing binary messages along edges connecting variable nodes and constraint nodes. Let $r_{i}\in\{0,1\}$ denote the received channel value for the $i$-th variable node, and let $\nu_{i,j}^{(\ell)}\in\{0,1\}$ be the binary message from the $i$-th variable node to the $j$-th constraint node in the $\ell$-th iteration. For simplicity, we assume no bit appears twice in a constraint, and let $\sigma_{j}(k)$ be the index of the variable node connected to the $k$-th socket of the $j$-th constraint. Let $j'$ be the other neighbor of the $i$-th variable node, and $\sigma_{j}(k)=i$. Then, the iterative decoder is defined by the recursion 
\[
\nu_{i,j'}^{(\ell+1)}=\iDec_{k}\left(\bm{v}_{i,j}^{(\ell)}\right),
\]
where the candidate decoding vector for the $j$-th constraint node and $i$-th variable node is
\begin{equation}
\!\!\!\bm{v}_{i,j}^{(\ell)}\!\triangleq\!\left(\!\nu_{\sigma_{j}(1),j}^{(\ell)},\!\cdots\!,\nu_{\sigma_{j}(k-1),j}^{(\ell)},r_{i},\nu_{\sigma_{j}(k+1),j}^{(\ell)},\!\cdots\!,\nu_{\sigma_{j}(n),j}^{(\ell)}\right)\!.\!\!\label{eq: Candidate}
\end{equation}
Note that the $k$-th entry is replaced by $r_{i}$. It is important to note that the above decoder passes extrinsic messages and \emph{is not identical} to the conventional approach that simply iterates by exchanging the outputs of the component code decoders. In particular, replacing the $k$-th element by the received channel output enables rigorous DE analysis. For the special case of Hamming component codes, this algorithm simplifies and becomes identical to Algorithm A in~\cite{Tanner-it81}. An illustrative figure showing the messages of the proposed iterative HDD on the graph can be found in Figure \ref{fig:IterHDD}. At first glance, it may not be clear that this algorithm allows a DE analysis. To see this more clearly, one can picture a different constraint-node operation where decoding is done twice, once with the $k$-th entry in~\eqref{eq: Candidate} replaced by~0 and once with it replaced by~1. Then, the variable node chooses the correct value based on $r_{i}$. This idea is used to describe the low-complexity version of this decoder in Section~\ref{subsec:LowComplexEMP}. 

\begin{figure*}[t]
\begin{center}
\scalebox{0.8}{%
\ifextpdf
\includegraphics{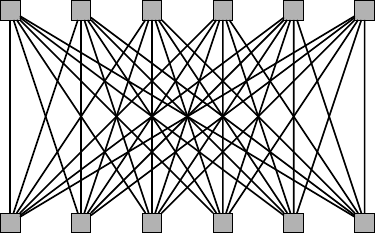}
\else
\begin{tikzpicture}[scale=1.2]
\tikzset{VertexStyle/.style ={fill=black!30,regular polygon,regular polygon sides=4,draw}}
\SetVertexMath
\SetVertexNoLabel
\grCompleteBipartite[RA=1,RB=1,RS=3,prefix=a,prefixx=b,Math=true]{6}{6}
\end{tikzpicture}\fi}
\hspace{10mm}
\scalebox{0.8}{%
\ifextpdf
\includegraphics{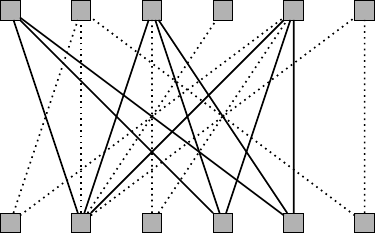}
\else
\begin{tikzpicture}[scale=1.2]
\tikzset{VertexStyle/.style ={fill=black!30,regular polygon,regular polygon sides=4,draw}}
\SetVertexMath
\SetVertexNoLabel
\grEmptyPath[Math,RA=1,RS=0,prefix=a]{6}   \grEmptyPath[Math,RA=1,RS=3,prefix=b]{6} 
\EdgeFromOneToSel{a}{b}{1}{0,2,4}
\EdgeFromOneToSel{a}{b}{3}{0,2,4}
\EdgeFromOneToSel{a}{b}{4}{0,2,4}   \EdgeFromOneToSel[style=dotted]{a}{b}{1}{1,3,5}   \EdgeFromOneToSel[style=dotted]{a}{b}{2}{2,4}   \EdgeFromOneToSel[style=dotted]{a}{b}{5}{1,5}   \EdgeFromOneToSel[style=dotted]{a}{b}{0}{1,4}  \end{tikzpicture}\fi}
\end{center}\vspace{-4mm}

\caption{On the left is the Tanner graph of a product code based on a $(6,3)$ binary code with minimum distance 3. The degree-2 bit nodes are suppressed and represented only by an edge. On the right is the ``error subgraph'' where each drawn edge (solid or dotted) denotes an erased bit. Since each code can correct $2$ erasures, decoding continues until all codes have either 0 erasures or $\geq3$ erasures. Thus, the $3$-core of this graph (denoted by solid edges) is equal to the stopping set that is found by iterative decoding. \label{fig:kcore}}
\end{figure*}

A \emph{stopping set} is an error pattern where the messages associated with every component code either have $0$ errors or greater than $t$ errors. For such a pattern, it is easy to verify that the conventional approach of running the ideal decoder for each component code results in no changes to the messages. Since each iteration of decoding can only reduce the number of errors, the final error pattern to which both decoders converge must be a stopping set. For the iterative HDD, we define the final error pattern as the set of bits where both component codes send a $1$ (i.e., error) message. This odd convention follows from the fact that stopping sets arise somewhat more naturally in the context of erasure channels and decoding. It also has the benefit that, with ideal component decoders, both the conventional approach and the above iterative HDD return the same final error pattern after sufficiently many iterations. As we will see later, this equivalence does not hold when the component code decoders introduce miscorrections.

This decoding problem is also very closely connected to a well-known greedy algorithm for finding the $k$-core in a graph~\cite{Pittel-jctb96}. The $k$-core is the largest induced subgraph where all vertices have degree at least $k$. The connection can be seen by considering an \emph{error graph} whose vertices are the code constraints where two vertices are connected if there is a bit in both code constraints and that bit is received error. One can obtain this error graph from the Tanner graph by deleting all variable nodes associated with correctly received bits and then collapsing the remaining degree-2 variables nodes into edges that connect two constraint nodes. Therefore, the degree of a constraint vertex is equal to the number of errors in its attached bits and the error graph represents all errors and their associations with component codes (e.g., see Figure~\ref{fig:kcore}). The greedy algorithm to find the $k$-core proceeds by removing any vertex of degree less than $k$ along with all of its edges (because this cannot possibly be part of the $k$-core). Likewise, conventional iterative decoding with ideal component decoders can be seen as correcting the errors in a component code that contains fewer than $t+1$ errors. Therefore, any stopping set found by iterative decoding without miscorrection is equivalent to the $(t+1)$-core of the error graph. 

\subsection{Iterative HDD with Bounded Distance Decoders\label{subsec: iterBDD}}

It is well-known that GLDPC codes perform well under iterative soft decoding \cite{Lodge-icc93,Pyndiah-com98}. The main drawback is that a posteriori probability (APP) decoding of the component codes can require significant computation. For this reason, we consider iterative HDD with bounded-distance decoding (BDD) of the component codes.%
{} Since the message update rule is the same as the rule introduced in Section \ref{subsec: IdealHDD}, the technique of DE can also be employed to analyze the performance of this algorithm. Moreover, a practical implementation of the iterative HDD algorithm is proposed in Section \ref{sec: LowComplex},

Likewise, we start by defining the bit-level mapping implied by BDD, denoted by $\Dec_{i}\colon\{0,1\}^{n}\to\{0,1\}$, which maps the received vector $\bm{v}\in\{0,1\}^{n}$ to the $i$-th decoded bit according to the rule \vspace{0mm}

\[
\Dec_{i}(\bm{v})\triangleq\begin{cases}
c_{i} & \mbox{if }\bm{c}\in\mathcal{C}\mbox{ satisfies }d_{H}(\bm{c},\bm{v})\leq t\\
v_{i} & \mbox{if }d_{H}(\bm{c},\bm{v})>t\mbox{ for all }\bm{c}\in\mathcal{C}.
\end{cases}
\]
It is easy to verify that this decoder satisfies the symmetry condition, \emph{i.e.,} $\Dec_{i}(\bm{v}\oplus\bm{c})=\Dec_{i}(\bm{v})\oplus c_{i}$ for all $\bm{c}\in\mathcal{C}$ and $i=1,\ldots,n$. %

We follow the same definition in Section \ref{subsec: IdealHDD}. Let $r_{i}\in\{0,1\}$ denote the received channel value for variable node $i$ and $\nu_{i,j}^{(\ell)}\in\{0,1\}$ be the binary message from the $i$-th variable node to the $j$-th constraint node in the $\ell$-th iteration. The iterative decoder is defined by the recursion 
\begin{equation}
\nu_{i,j'}^{(\ell+1)}=\Dec_{k}\left(\bm{v}_{i,j}^{(\ell)}\right),\label{eq:MessagePassingRule}
\end{equation}
where%
{} $\vect v_{i,j}^{(\ell)}$ is the candidate decoding vector for the $j$-th constraint node and the $i$-th variable node defined in (\ref{eq: Candidate}). The corresponding DE analysis is discussed in Section \ref{subsec: BDDDE}.

\section{Summary of Main Results}

Let BSC$(p)$ denote a BSC with error probability $p$. According to the channel coding theorem, a long random code can achieve reliable communication over a BSC$(p)$ if its code rate is less than $C(p)$, where $C(p)$ is the channel capacity of a BSC$(p)$. Now, consider a code ensemble of rate $R$ that has a noise threshold of $p^{*}$. Since $1-R$ is the normalized redundancy for the code, one can interpret the quantity $1-C(p)$ as the minimal normalized redundancy required for reliable communication. We evaluate the efficiency of the code ensemble by comparing its normalized redundancy with $1-C(p^{*})$. In this paper, we say that a code ensemble can approach capacity if the ratio $\frac{1-C(p^{*})}{1-R}$ can be made arbitrarily close to $1$. A formal statement is given by the following definition.
\begin{defn}
\label{def: EpsRedund}Let $C(p)$ be the capacity of a BSC$(p)$. For some $\epsilon>0$, a code ensemble with rate $R$ and threshold $p^{*}$ is called \emph{$\epsilon$-redundancy achieving} if 
\begin{align*}
\frac{1-C\left(p^{*}\right)}{1-R} & \geq1-\epsilon.
\end{align*}
\end{defn}
Using Definition \ref{def: EpsRedund}, the main result of this paper is the following theorem.
\begin{thm}
\label{thm: main}For any $\epsilon>0$, there exists a tuple $(t,n,L,w)$ such that iterative HDD of the $(\set C,m,L,w)$ GLDPC spatially-coupled ensemble is $\epsilon$-redundancy achieving when $\mathcal{C}$ is a $t$-error correcting BCH code of length $n$.
\end{thm}
\begin{IEEEproof}
See Section \ref{subsec: MainProof}.
\end{IEEEproof}
\begin{rem}
The proof of Theorem~\ref{thm: main} constructs iterative coding systems whose $\epsilon$-redundancy can be made arbitrarily small. However, as $\epsilon\to0$, the code rate approaches 1 and the system parameters satisfy $m\gg L\gg w\gg t$ and $t\to\infty$.
\end{rem}
In the rest of this paper, we prove Theorem \ref{thm: main} using the following steps. The DE analysis of the proposed iterative HDD algorithms is introduced in Section \ref{sec: DE}. The high-rate scaling limits of the DE updates, for both $(\Cc,m)$ and $(\Cc,m,L,w)$ GLDPC ensembles with BCH component codes, are described in Section \ref{sec:BCH}. The thresholds of the iterative HDD, for both $(\Cc,m)$ and $(\Cc,m,L,w)$ GLDPC ensembles, are analyzed in Section \ref{sec:Bounds}. Using the threshold results in Section \ref{sec:Bounds}, the proof of Theorem \ref{thm: main} is provided in Section \ref{sec: Approaching}.%

\section{Density Evolution \label{sec: DE}}

By the fact that any fixed-depth neighborhood of a randomly chosen vertex in the Tanner graph is a tree with high probability as $m\rightarrow\infty$ \cite[\S 3.8]{RU-2008}, and from the messages shown in Figure \ref{fig:IterHDD}, one can easily expand a depth-$m$ tree and show that the inputs to the tree are all from received channel values. Under the assumption that the channel values are independent and identically distributed random variables, the density of the messages passed on the edges can be precisely computed using a DE recursion. For HDD of the component codes, this DE can be written as a one-dimensional recursion. 

\subsection{Iterative HDD with Ideal Component Decoders\label{subsec: IdealDE}}

For a randomly chosen edge $(i,j')$ connecting a variable node $i$ and a constraint node $j'$, let $j\triangleq\nbr(i)\setminus j'$ be the other neighbor of the bit node $i$. One can see that the message passed on $(i,j')$ is an error only when the channel output of the $i$-th bit is an error and there are at least $t$ error inputs to the constraint node $j$. For a $n\in\mathbb{N}$, let $x^{(\ell)}$ be the error probability of a random chosen bit-to-constraint message in the $\ell$-th iteration. The DE recursion of the iterative HDD with ideal component decoders for the $\left(\mathcal{C},m\right)$ GLDPC ensemble is%
{} 
\begin{equation}
x^{(\ell+1)}=p\hat{f}_{n}\left(x^{(\ell)}\right),\label{eq: NomissDE}
\end{equation}
where 
\begin{align}
\hat{f}_{n}(x) & \triangleq\sum_{i=t}^{n-1}\binom{n-1}{i}x^{i}\big(1-x\big)^{n-i-1}.\label{eq: hfn}
\end{align}

Let the noise threshold of the iterative HDD with ideal component decoders be defined by 
\begin{align*}
\hpn^{*} & \triangleq\sup\left\{ p\in(0,1]\,\big|\,p\hfn(x)<x,\,x\in(0,p]\right\} .
\end{align*}
Similar to DE for LDPC codes on the BEC \cite[pp. 95--96]{RU-2008}, there is a compact characterization of the hard-decision decoding threshold $\hpn^{*}$.%
{} Since $p\hfn(x)$ is monotone in $p$, the threshold $\hpn^{*}$ can be obtained by 
\begin{equation}
\hpn^{*}=\inf_{x\in(0,1)}\frac{x}{\hf_{n}(x)}.\label{eq: BPthreshold}
\end{equation}

\begin{rem}
This type of analysis is also related to the threshold analysis for the $k$-core problem in~\cite{Pittel-jctb96}. Schwartz \emph{et al.} also perform a combinatorial analysis in~\cite{Schwartz-isit05} to determine the decoding threshold for asymptotically long product codes. Their conclusion is somewhat different from other reported results because they assume a finite number of decoding iterations and require that the block error rate vanishes. However, they treat the number of iterations explicitly and one can extract threshold estimates from~\cite[Cor.~2]{Schwartz-isit05} in the limit as the number decoding rounds tends to infinity. In this case, a little algebra shows that the threshold $c$-value is $c^{*}=(t!)^{1/t}$. If the number of decoding rounds is greater than $2\log\log n$ as $n\to\infty$, then their equations imply that the noise threshold is $p^{*}=c^{*}/n=(t!)^{1/t}/n$. However, their analysis does not allow the number of iterations to depend on $n$. Thus, this calculation only gives a lower bound on the correct threshold. 
\end{rem}
For the $(\Cc,m,L,w)$ ensemble, let $x_{i}^{(\ell)}$ be the average error probability of the hard-decision messages emitted by bit nodes at position $i$ in the $\ell$-th iteration. Assume that $x_{i}^{(\ell)}=0$ for all $i\notin\{1,2,\dots,L\}$ and $\ell\geq0$. According to the construction in Definition \ref{def: SCEnsemble}, the average error probability of the hard-decision inputs to code-constraint nodes at position $i$ is $y_{i}^{(\ell)}=\frac{1}{w}\sum_{j=0}^{w-1}x_{i-j}^{(\ell)}$. Then, the error probability of the hard-decision messages emitted by bit nodes at position $i$ in the $(\ell+1)$-th iteration is 
\begin{align}
x_{i}^{(\ell+1)} & =\frac{1}{w}\sum_{k=0}^{w-1}p\hf_{n}\left(y_{i+k}^{(\ell)}\right)\nonumber \\
 & =p\left(\frac{1}{w}\sum_{k=0}^{w-1}\hat{f}_{n}\left(\frac{1}{w}\sum_{j=0}^{w-1}x_{i-j+k}^{(\ell)}\right)\right).\label{eq: SCDE}
\end{align}

\begin{rem}
\label{rem: ConstraintInput} An LDPC ensemble whose DE results in spatial averaging, similar to~(\ref{eq: SCDE}), was introduced by Kudekar \emph{et al.} in~\cite{Kudekar-it11}. In this case,~(\ref{eq: SCDE}) tracks the average error probability of the output hard-decision messages from bit nodes at each position. One can also obtain the DE update of the average error probability of the input hard-decision messages to the code-constraint nodes at the position $i\in\{1,2,\dots,L+w-1\}$ by
\begin{align}
y_{i}^{(\ell+1)} & =\frac{1}{w}\sum_{j=\max\{i-L,0\}}^{\min\{i-1,w-1\}}x_{i-j}^{(\ell+1)}\nonumber \\
 & =\frac{1}{w}\sum_{j=\max\{i-L,0\}}^{\min\{i-1,w-1\}}p\left(\frac{1}{w}\sum_{k=0}^{w-1}\hf_{n}\left(y_{i-j+k}^{(\ell)}\right)\right).\label{eq: SCDE2}
\end{align}
In the following analysis, we use~(\ref{eq: SCDE2}) to find the noise threshold of the spatially-coupled system with iterative HDD because this is the update used in~\cite{Yedla-istc12,Yedla-it14}. The thresholds obtained from~(\ref{eq: SCDE}) and from~(\ref{eq: SCDE2}) are identical because, for any finite $w>0$, it is easy to verify that $x_{i}^{(\ell)}\rightarrow0$ for all $i\in\{1,2,\dots,L\}$ if and only if $y_{i}^{(\ell)}\rightarrow0$ for all $i\in\{0,1,\dots,L+w-1\}$, .
\end{rem}

\subsection{Iterative HDD with BDD\label{subsec: BDDDE}}

Since the component decoder is symmetric, it suffices to consider the case where the all-zero codeword is transmitted over a BSC with error probability $p$ \cite[pp. 188--191]{RU-2008}. Let $x^{(\ell)}$ be the error probability of the hard-decision messages passed from the variable nodes to the constraint nodes after $\ell$ iterations. For an arbitrary symmetric decoder, let $P_{n}(i)$ be the probability that a randomly chosen bit is decoded incorrectly when it is initially incorrect and there are $i$ random errors in the other $n-1$ inputs. Likewise, let $Q_{n}(i)$ be the probability that a randomly chosen bit is decoded incorrectly when it is initially correct and there are $i$ random errors in the other $n-1$ inputs. Then, for the $\left(\mathcal{C},m\right)$ GLDPC ensemble, the DE recursion implied by (\ref{eq:MessagePassingRule}) is defined by $x^{(0)}=p$, 
\begin{equation}
x^{(\ell+1)}=f_{n}\left(x^{(\ell)};p\right),\label{eq: fnDE}
\end{equation}
and (with $\overline{p}\triangleq1-p$) 
\begin{equation}
\!\!\!f_{n}(x;p)\!\triangleq\!\sum_{i=0}^{n-1}\!\binom{n\!-\!1}{i}x^{i}\big(1\!-\!x\big)^{n-i-1}\!\left(pP_{n}(i)\!+\!\overline{p}Q_{n}(i)\right)\!.\!\!\label{eq: fn}
\end{equation}

\begin{rem}
According to the definition of $P_{n}(i)$, it is clear that $P_{n}(i)=0$ for all $0\leq i\leq t-1$ since the total number of errors is less than or equal to $t$ and all errors are corrected by BDD. When $i\geq t$, $P_{n}(i)$ is the probability that both the input and the output of a randomly selected bit are errors. Given that the input of a randomly selected bit is an error, there are two cases where the output of the randomly selected bit is an error. The first case is that there is no codeword is within a distance of $t$ from the inputs. In this case, the error input of the randomly selected bit is passed to the output. The second case is that the decoder returns a codeword but the randomly selected bit is still an error. Since these cases are disjoint events, $P_{n}(i)$ is the sum of the probabilities of these two events. 

Similarly, it is easy to show that $Q_{n}(i)=0$ for $0\leq i\leq t$. Since the input of the randomly selected bit is always correct in this case, the only error event for $Q_{n}(i)$ is that a codeword with an error at the randomly selected bit is returned. Consider BDD with fixed decoding radius $t$, and a fixed $i\geq t$. It is easy to show that the number of error vectors with $i+1$ errors increases as $n$ increases. Since the decoding radius is fixed to $t$, the probability that a randomly selected error vector of weight $i+1$ falls in a $t$-ball centered at a codeword decreases as $n$ increases. Thus, $P_{n}(i)$ is increasing with respect to $n$. For an $(n,k,2t+1)$ binary primitive BCH code (or its $(n,k-1,2t+2)$ even-weight subcode), a rigorous proof showing $\lim_{n\to\infty}P_{n}(i)=1$ for $t\leq i\leq n-t-2$ is introduced in Lemma~\ref{lem:LimitOfPn(i)andnQn(i)}. \label{rem: P_n(i)Q_n(i)}
\end{rem}
For the iterative HDD with BDD described above, the quantities $P_{n}(i)$ and $Q_{n}(i)$ can be written in terms of the number of codewords of weight $l$ in $\set C$, denoted by $A_{l}$, \cite{Kim-com96}. Using the convention that $\binom{n}{k}=0$ if $n<0$, $k<0$, or $k>n$, we define 
\begin{equation}
l(i,\delta,j)\triangleq i-\delta+2j+1,\label{eq: l(i,d,j)}
\end{equation}
 
\[
\itTheta(n,i,\delta,j)\!\triangleq\!\binom{l(i,\delta,j)}{l(i,\delta,j)\!-\!j}\!\binom{n\!-\!l(i,\delta,j)\!-\!1}{\delta-1-j}\!\binom{n\!-\!1}{i}^{-1},
\]
and
\begin{align*}
\itLambda(n,i,\delta,j) & \!\triangleq\!\binom{l(i,\delta,j)\!-\!2}{l(i,\delta,j)\!-\!j\!-\!1}\!\binom{n\!-\!l(i,\delta,j)\!+\!1}{\delta-j}\!\binom{n\!-\!1}{i}^{-1}\!\!.
\end{align*}
Since all decoding regions are disjoint, one can compute %
\begin{align}
P_{n}(i) & =1-\sum_{\delta=1}^{t}\sum_{j=0}^{\delta-1}\frac{n-l(i,\delta,j)}{n}A_{l(i,\delta,j)}\itTheta(n,i,\delta,j)\label{eq: P(i)}
\end{align}
for $t\!\leq\!i\!\leq\!n\!-\!t\!-\!2$ and $P_{n}(i)\!=\!0$ for $0\!\leq\!i\!\leq\!t\!-\!1$. Similarly\footnote{The expression for $Q_{n}(i)$ in (\ref{eq: Q(i)}) corrects a small error in our previous work \cite[Eqn.~(4)]{Jian-isit12}.}, 
\begin{align}
Q_{n}(i) & =\sum_{\delta=1}^{t}\sum_{j=0}^{\delta}\frac{l(i,\delta,j)-1}{n}A_{l(i,\delta,j)-1}\itLambda(n,i,\delta,j)\label{eq: Q(i)}
\end{align}
for $t+1\leq i\leq n-t-1$ and $Q_{n}(i)=0$ for $0\leq i\leq t$. Note that, when the code contains the all-one codeword, $P_{n}(i)=1$ for $n-t-1\leq i\leq n-1$, and $Q_{n}(i)=1$ for $n-t\leq i\leq n-1$.

Let the noise threshold be defined by 
\begin{align}
p_{n}^{*} & \triangleq\sup\left\{ p\in(0,1]\,\big|\,f_{n}(x;p)<x,\,x\in(0,p]\right\} .\label{eq: pn}
\end{align}
According to the definition of $p_{n}^{*}$, one can show that $x^{(\ell)}\rightarrow0$ as $\ell\rightarrow\infty$ for all $p<p_{n}^{*}$. By rewriting $f_{n}(x;p)$ as $f_{n}(x;p)=p[f_{n}(x;1)-f_{n}(x;0)]+f_{n}(x;0)$, we know that $f_{n}(x;p)$ is monotone in $p$. However, it is not clear to us whether or not $f_{n}(x;p)$ is monotone in $x$. Therefore, the definition of $p_{n}^{*}$ does not imply that $\lim_{\ell\rightarrow\infty}x^{(\ell)}>0$ for all $p>p_{n}^{*}$. Define $x^{*}\triangleq\sup\{z\in[0,1]\mid f_{n}(x;0)\leq x\mbox{ for all }x\in(0,z]\}$. We can characterize $p_{n}^{*}$ similar to (\ref{eq: BPthreshold}) by 
\[
p_{n}^{*}=\inf_{x\in(0,x^{*})}\frac{x-f_{n}(x;0)}{f_{n}(x;1)-f_{n}(x;0)}.
\]
The idea is that each fixed-point value $x\in(0,x^{*})$ provides a upper bound on the threshold $p_{n}^{*}$ and that infimum of these lower bounds gives $p_{n}^{*}$ because there is some fixed-point value that is achieved at the threshold.
\begin{rem}
Since the operations at bit nodes and constraint nodes are both sub-optimal, one may observe that for a fixed tuple $(p,n,t)$, there exist some $x\in[0,1]$ such that $f_{n}(x;p)>p$. This implies that the average error probability of the messages emitted by bit nodes after one iteration will be worse than the error probability of the channel output. In this case, bit nodes can just send the received channel bits to their neighbors. With this modification, the resulting DE update equation is
\begin{align*}
x^{(\ell+1)} & =\min\{p,f_{n}(x^{(\ell)};p)\}.
\end{align*}
Let $q_{n}^{*}\triangleq\sup\left\{ p\in(0,1]\,\big|\,\min\left\{ f_{n}(x;p),p\right\} <x,\,x\in(0,p]\right\} $ be the noise threshold of the modified decoding algorithm. We claim that $p_{n}^{*}=q_{n}^{*}$ by the following argument. From the fact that $\min\{p,f_{n}(x;p)\}\leq f_{n}(x;p)$, we have $q_{n}^{*}\geq p_{n}^{*}$. Consider the case of $p>p_{n}^{*}$. From (\ref{eq: pn}), there exists some $x_{0}\in(0,p]$ such that $f_{n}(x_{0};p)\geq x_{0}$. Since $x_{0}\in(0,p]$, one can show that $\min\{p,f_{n}(x_{0};p)\}\geq x_{0}$ as well. Thus, we know $p>q_{n}^{*}$ from the definition of $q_{n}^{*}$. This implies that $q_{n}^{*}\leq p_{n}^{*}$, and therefore we conclude that $p_{n}^{*}=q_{n}^{*}$. Also, this change seems to have no significant effect on the number of iterations required to achieve a certain error rate.
\end{rem}

To derive the DE update equation of the $\left(\mathcal{C},m,L,w\right)$ spatially-coupled GLDPC ensemble, let $x_{i}^{(\ell)}$ be the average error probability of hard-decision messages emitted by bit nodes at position $i$ after the $\ell$-th iteration. According to Definition \ref{def: SCEnsemble}, the average error probability of input messages to a code-constraint node at position $i$ is $y_{i}^{(\ell)}=\frac{1}{w}\sum_{j=0}^{w-1}x_{i-j}^{(\ell)}$. It follows that $x_{i}^{(\ell+1)}=\frac{1}{w}\sum_{k=0}^{w-1}f_{n}(y_{i+k}^{(\ell)};p)$ for $i\in\{1,2,\dots,L\}$, where $f_{n}(x;p)$ is defined in (\ref{eq: fn}). We also set $x_{i}^{(\ell)}=0$ for $i\notin\{1,2,\dots,L\}$. Therefore, in the $(\ell+1)$-th iteration, the average error probability of the hard-decision messages emitted by bit nodes at the position $i\in\{1,2,\dots,L\}$ is given by \vspace{0mm}
\begin{equation}
x_{i}^{(\ell+1)}=\frac{1}{w}\sum_{k=0}^{w-1}f_{n}\left(\frac{1}{w}\sum_{j=0}^{w-1}x_{i-j+k}^{(\ell)};p\right).\label{eq: SCDE3}
\end{equation}
Similar to the discussion in Remark \ref{rem: ConstraintInput}, the DE update of the $\left(\mathcal{C},m,L,w\right)$ spatially-coupled GLDPC ensemble which tracks the average error probability of the input messages to the constraint nodes at the position $i\in\{1,2,\dots,L+w-1\}$ can also be written as
\begin{align*}
y_{i}^{(\ell+1)} & =\frac{1}{w}\sum_{j=\max\{i-L,0\}}^{\min\{i-1,w-1\}}x_{i-j}^{(\ell+1)}\\
 & =\frac{1}{w}\sum_{j=\max\{i-L,0\}}^{\min\{i-1,w-1\}}\frac{1}{w}\sum_{k=0}^{w-1}f_{n}\left(y_{i-j+k}^{(\ell)};p\right).
\end{align*}

\vspace{0mm}

\section{BCH Component Codes\label{sec:BCH}}

In the remainder of this paper, an $(n,k,2t+1)$ binary primitive BCH code (or its $(n,k-1,2t+2)$ even-weight subcode) will be used as the component code for both the $(\set C,m)$ GLDPC and $(\set C,m,L,w)$ spatially-coupled GLDPC ensembles. When the exact weight spectrum is known, one can compute $P_{n}(i)$ and $Q_{n}(i)$ using (\ref{eq: P(i)}) and (\ref{eq: Q(i)}), respectively. Otherwise, we use the asymptotically-tight binomial approximation 
\begin{align}
A_{l} & =\begin{cases}
2^{-\nu t}\binom{n}{l}\left(1+O\left(n^{-0.1}\right)\right) & \text{if }d\leq l\leq n-d,\\
1, & \text{if }l=0,\ l=n,\\
0, & \mbox{otherwise},
\end{cases}\label{eq: WeightDistrib}
\end{align}
for $n\geq n_{t}$, where $d=2t+1$, $n=2^{\nu}-1$ and $n_{t}$ is a constant depends on $t$ \cite{Sidelnokiv-ppi71}. %

For the $(n,k-1,2t+2)$ even-weight subcode of an $(n,k,2t+1)$ primitive BCH code, the number of codewords is denoted by $\tilde{A}_{l}$ where $\tilde{A}_{l}=A_{l}$ when $l$ is even and $\tilde{A}_{l}=0$ when $l$ is odd. Let $\tilde{P}_{n}(i)$ and $\tilde{Q}_{n}(i)$ be the miscorrection probabilities implied by $\tilde{A}_{l}$ for the even-weight subcode. Similar to $P_{n}(i)$ and $Q_{n}(i)$ in the $(n,k,2t+1)$ primitive BCH code, it can be shown that $\tilde{P}_{n}(i)=0$ for $0\leq i\leq t-1$ and $\tilde{Q}_{n}(i)=0$ for $0\leq i\leq t+1$. Then, the DE recursions for the $(\set C,m)$ GLDPC ensemble and the $(\set C,m,L,w)$ spatially-coupled GLDPC ensemble can be obtained from (\ref{eq: fnDE}) and (\ref{eq: SCDE3}), respectively. %

\subsection{Notation}

The following glossary of symbols is designed to help readers distinguish between closely related symbols:
\begin{description}
\item [{$p\hfn(x)$}] DE update for ideal decoding
\item [{$\hpn^{*}$}] DE threshold for ideal decoding 
\item [{$f_{n}(x;p)$}] DE update for BDD 
\item [{$p_{n}^{*}$}] DE threshold for BDD 
\item [{$\hat{\rho}^{*}$}] Poisson DE threshold for ideal decoding 
\item [{$f(\lambda^{(\ell)};\rho)$}] Poisson DE update for BDD
\item [{$\rho^{*}$}] Poisson DE threshold for BDD
\item [{$\hV_{n}(x;p)$}] potential function for ideal decoding
\item [{$\hpn^{**}$}] SC-DE threshold for ideal decoding
\item [{$\hV_{n}(x)$}] fixed-point potential for ideal decoding
\item [{$\hU(x;p)$}] potential function for combined ideal system
\item [{$\hV(\lambda;\rho)$}] potential for ideal Poisson DE 
\item [{$\hat{\rho}_{t}^{**}$}] SC-DE threshold for ideal Poisson DE 
\item [{$U(\lambda;\rho)$}] potential for Poisson DE with BDD
\item [{$\rho_{t}^{**}$}] SC-DE threshold for Poisson DE with BDD
\item [{$p_{n}^{**}$}] SC-DE threshold for BDD
\end{description}

\subsection{High-Rate Scaling Limit for Iterative HDD with Ideal Component Decoders}

In \cite{Justesen-itw07,Justesen-toc11}, Justesen \emph{et al. }analyze the asymptotic performance of long product codes under the assumption that the component decoders never miscorrect. These arguments can be applied for the decoding of both BSC and BEC outputs.%
{} By considering the decoding process as removing vertices of degree less or equal to $t$, they show that the process fails if the error graph contains $(t+1)$-core. The existence of a ``$k$-cores'' in a random graph has attracted considerable interest in graph theory \cite{Pittel-jctb96}. By employing the results in \cite{Pittel-jctb96}, Justesen \emph{et al.} characterize the evolution for the number of errors per constraint node as a recursion for the ``Poisson parameter''~\cite{Justesen-itw07,Justesen-toc11}. That recursion leads to a threshold, for successful decoding, on the average number of error bits attached to a code-constraint node. 

In terms of error rate, this threshold has the form $p=c/n$, where $n$ is the length of the component code. Thus, the error rate vanishes as the size of the product code increases. In this section, we consider the DE recursion for the iterative decoding of the GLDPC and SC-GLDPC ensembles when the BSC error rate scales like $c/n$. We refer to this as the \emph{high-rate scaling} regime.

Now, we will derive the limiting DE recursion in the high-rate scaling regime for the proposed iterative algorithm with ideal HDD. Then, we observe that the obtained high-rate scaling limit (\ref{eq: HighRateNoMiss}) has the same update equation as the recursion of Poisson parameter in \cite{Justesen-itw07,Justesen-toc11}.

For a fixed $\rho>0$, let $p\triangleq\frac{\rho}{n-1}$ scale with $n$ and $\lambda_{n}^{(\ell)}\triangleq(n-1)x^{(\ell)}$. The recursion (\ref{eq: NomissDE}) for $\lambda_{n}^{(\ell)}$ becomes 
\[
\lambda_{n}^{(\ell+1)}=(n-1)p\hfn\left(\frac{\lambda_{n}^{(\ell)}}{n-1}\right)=\rho\hfn\left(\frac{\lambda_{n}^{(\ell)}}{n-1}\right)
\]
starting from $\lambda_{n}^{(0)}=\rho$. For $\ell>0$, define $\lambda^{(\ell)}\triangleq\lim_{n\rightarrow\infty}\lambda_{n}^{(\ell)}$ and observe that the high-rate scaling limit of the recursion for the ideal component code decoder is
\begin{align*}
\lambda^{(\ell+1)} & =\hf\left(\lambda^{(\ell)};\rho\right)\triangleq\lim_{n\rightarrow\infty}\rho\hfn\left(\frac{\lambda_{n}^{(\ell)}}{n-1}\right).
\end{align*}
The existence of the limit can be established via induction on $\ell$. To see this, define the tail probability of the Poisson distribution with mean $\lambda$ by 
\begin{align*}
\phi(\lambda;k) & \triangleq\sum_{i=k+1}^{\infty}\frac{\lambda^{i}}{i!}e^{-\lambda}.
\end{align*}
Then, the Poisson theorem \cite[pp. 113]{Papoulis-2002} shows that $\lim_{n\rightarrow\infty}\rho\hfn(\frac{\lambda}{n-1})=\rho\phi\left(\lambda;t-1\right)$. Thus, the high-rate scaling limit of the recursion for the ideal component code decoder becomes 
\begin{align}
\lambda^{(\ell+1)} & =\rho\phi\left(\lambda^{(\ell)};t-1\right).\label{eq: HighRateNoMiss}
\end{align}
The scaled noise threshold $\hat{\rho}^{*}$ is the largest $\rho$ such that the iteration converges from $\rho$ to 0 and is defined by
\[
\hat{\rho}^{*}\triangleq\sup\left\{ \rho\in[0,\infty)\,\big|\,\rho\phi\left(\lambda;t-1\right)<\lambda,\,\lambda\in(0,\rho]\right\} .
\]
Since $\phi(\lambda;t-1)$ is increasing and upper bounded by 1, it follows that
\[
\hat{\rho}^{*}=\inf_{\lambda>0}\frac{\lambda}{\phi\left(\lambda;t-1\right)}
\]

\begin{rem}
This threshold condition is identical to the ones given in \cite{Pittel-jctb96} for the equivalent $k$-core problem. The connection to asymptotically long product codes with ideal bounded distance decoding was first made in~\cite{Justesen-itw07}.
\end{rem}
For the spatially-coupled GLDPC ensemble, let $\lambda_{i}^{(\ell)}$ with $i\in\{1,2,\dots,L\}$ be the average number of error messages emitted by bit nodes at position $i$ in the $\ell$-th iteration. We set $\lambda_{i}^{(0)}=\rho$ for all $i\in\{1,2,\dots,L\}$, and set $\lambda_{i}^{(\ell)}=0$ for all $i\notin\{1,2,\dots,L\}$ and $\ell\geq0$. Very similar same arguments show that the recursion for the spatially-coupled ensemble is 
\begin{align}
\lambda_{i}^{(\ell+1)} & =\frac{1}{w}\sum_{k=0}^{w-1}\hf\left(\frac{1}{w}\sum_{j=0}^{w-1}\lambda_{i-j+k}^{(\ell)};\rho\right)\nonumber \\
 & =\rho\left(\frac{1}{w}\sum_{k=0}^{w-1}\phi\left(\frac{1}{w}\sum_{j=0}^{w-1}\lambda_{i-j+k}^{(\ell)};t-1\right)\right).\label{eq: SCPoiDENoMiss}
\end{align}

\begin{rem}
We note that this vector update equation is the same as the vector update equation used for $Q$-coloring analysis (with $Q=t+1$) of spatially-coupled graphs in~\cite{Hassani-jsp13}. Therefore, the threshold bounds in Section~\ref{subsec:IdealIterativeHDD} also apply to the $Q$-coloring problem.
\end{rem}

\subsection{High-Rate Scaling Limit for Iterative HDD with BDD}

We have shown that the recursion using the random graph argument in \cite{Justesen-itw07} and \cite{Justesen-toc11} is the same as the DE analysis of ideal iterative HDD in the limit as $n\rightarrow\infty$. The main weakness of the random graph argument is that it is not applicable to decoders with miscorrection. In this section, the high-rate scaling limit of the recursion for our DE analysis as $n\to\infty$ is introduced. The main contribution is that our approach rigorously accounts for miscorrection. 

First, we introduce some notation and a few lemmas to simplify the development. Consider the Poisson distribution with mean $\lambda$. Let $\psi(\lambda;k)$ and $\varphi(\lambda;k)$ be, respectively, the tail probability for the even terms, and the tail probability for the odd terms. Then, we have 
\begin{align*}
\psi(\lambda;k) & \triangleq\frac{1+e^{-2\lambda}}{2}-\sum_{i=0}^{\left\lfloor k/2\right\rfloor }\frac{\lambda^{2i}}{(2i)!}e^{-\lambda},\\
\varphi(\lambda;k) & \triangleq\frac{1-e^{-2\lambda}}{2}-\sum_{i=0}^{\left\lfloor k/2\right\rfloor }\frac{\lambda^{(2i+1)}}{(2i+1)!}e^{-\lambda}.
\end{align*}

\begin{lem}
\label{lem:LimitOfPn(i)andnQn(i)} For the codes described above and $t\leq i\leq n-t-2$,%
{} the limit $\lim_{n\to\infty}P_{n}(i)=1$. Also, for the same code and $t+1\leq i\leq n-t-1$, the function $nQ_{n}(i)$ is bounded. If $\lfloor\sqrt{n}\rfloor>t+1$, then 
\begin{align}
nQ_{n}(i) & \leq\frac{1}{(t-1)!}+O\left(n^{-0.1}\right)\label{eq: finite_n}
\end{align}
 for all $t+1\leq i\leq\lfloor\sqrt{n}\rfloor$. Thus, for any fixed $i\geq t+1$, we have $\lim_{n\rightarrow\infty}nQ_{n}(i)=\frac{1}{(t-1)!}$.
\end{lem}
\begin{IEEEproof}
See Appendix \ref{app: pfLemLimOfPnnQn}.
\end{IEEEproof}

Consider the DE recursion (\ref{eq: fnDE}) for the $(\set C,m)$ GLDPC ensemble. For a fixed $\rho$, let $p\triangleq\frac{\rho}{n-1}$ scale with $n$ and $\lambda_{n}^{(\ell)}\triangleq(n-1)x^{(\ell)}$. From (\ref{eq: fnDE}) and (\ref{eq: fn}), the recursion for $\lambda_{n}^{(\ell)}$ equals\vspace{0mm}
\begin{align}
\lambda{}_{n}^{(\ell+1)} & =(n-1)f_{n}\left(\frac{\lambda_{n}^{(\ell)}}{n-1};\frac{\rho}{n-1}\right)\nonumber \\
 & =\sum_{i=t}^{n-1}\binom{n-1}{i}\!\left(\frac{\lambda_{n}^{(\ell)}}{n-1}\right)^{i}\!\left(1-\frac{\lambda_{n}^{(\ell)}}{n-1}\right)^{n-1-i}\nonumber \\
 & \qquad\times\left(\rho\left(P_{n}(i)-Q_{n}(i)\right)\!+\!(n-1)\,Q_{n}(i)\right),\label{eq: ScaleDE2}
\end{align}
with initial value $\lambda_{n}^{(0)}=\rho$ for all $n$.%

\begin{lem}
\label{lem: limEPn_limEQn}Let $X_{n}\sim\mbox{Bi}(n-1,\frac{\lambda_{n}}{n-1})$ be a sequence of binomial random variables associated with $n-1$ trials with success probability $\frac{\lambda_{n}}{n-1}$. If $\lambda_{n}\rightarrow\lambda<\infty$, then $\lim_{n\rightarrow\infty}E\left[P_{n}(X_{n})\right]=\phi\left(\lambda;t-1\right)$ and $\lim_{n\rightarrow\infty}E\left[Q_{n}(X_{n})\right]=0$.
\end{lem}
\begin{IEEEproof}
See Appendix \ref{app: pf_limEPn_limEQn}.
\end{IEEEproof}
\begin{lem}
\label{lem: limEnQn}Let $X_{n}\sim\mbox{Bi}(n-1,\frac{\lambda_{n}}{n-1})$ be a sequence of binomial random variables for $n-1$ trials with success probability $\frac{\lambda_{n}}{n-1}$. If $\lambda_{n}\rightarrow\lambda<\infty$, then 
\begin{align*}
\lim_{n\rightarrow\infty}E\left[nQ_{n}(X_{n})\right] & =\frac{\phi\left(\lambda;t\right)}{(t-1)!}.
\end{align*}
\end{lem}
\begin{IEEEproof}
See Appendix \ref{app: pf_limEnQn}.
\end{IEEEproof}

Using Lemma \ref{lem: limEPn_limEQn} and Lemma \ref{lem: limEnQn}, one can simplify the recursion for $\lambda^{(\ell)}\triangleq\lim_{n\rightarrow\infty}\lambda_{n}^{(\ell)}$.
\begin{lem}
For any fixed $\ell>0$, the limit $\lambda^{(\ell)}\triangleq\lim_{n\rightarrow\infty}\lambda_{n}^{(\ell)}$ exists, and the recursion for $\lambda^{(\ell)}$ is given by $\lambda^{(0)}=\rho$ and %
\begin{align}
\lambda^{(\ell+1)} & =f\left(\lambda^{(\ell)};\rho\right)\nonumber \\
 & \triangleq\rho\phi\left(\lambda^{(\ell)};t-1\right)+\frac{1}{(t-1)!}\phi\left(\lambda^{(\ell)};t\right).\label{eq: PoiDE}
\end{align}
\end{lem}
\begin{IEEEproof}
We prove the lemma by induction.%
{} The base case $\lambda^{(0)}=\rho$ holds by assumption. For the inductive step, suppose that $\lambda^{(\ell)}=\lim_{n\rightarrow\infty}\lambda_{n}^{(\ell)}$ exists. Let $X_{n}^{(\ell)}\sim\mathrm{Bi}(n-1,\frac{\lambda_{n}^{(\ell)}}{n-1})$ be a binomial random variable with parameters $n-1$ and $\frac{\lambda_{n}^{(\ell)}}{n-1}$. Then, the recursion (\ref{eq: ScaleDE2}) can be represented as $\lambda_{n}^{(\ell+1)}=E[\rho(P_{n}(X_{n}^{(\ell)})-Q_{n}(X_{n}^{(\ell)}))+(n-1)Q_{n}(X_{n}^{(\ell)})]$. Again, Lemma \ref{lem: limEPn_limEQn} and Lemma \ref{lem: limEnQn} imply that $\lim_{n\rightarrow\infty}E[P_{n}(X_{n}^{(\ell)})]$, $\lim_{n\rightarrow\infty}E[Q_{n}(X_{n}^{(\ell)})]$, and $\lim_{n\rightarrow\infty}E[nQ_{n}(X_{n}^{(\ell)})]$ exist. Thus, the limit of $\lambda_{n}^{(\ell+1)}$ exists as $n\rightarrow\infty$, and satisfies the recursion 
\begin{align*}
\lambda^{(\ell+1)} & \triangleq\lim_{n\rightarrow\infty}\lambda_{n}^{(\ell+1)}=\rho\phi\left(\lambda^{(\ell)};t\!-\!1\right)\!+\!\frac{1}{(t\!-\!1)!}\phi\!\left(\!\lambda^{(\ell)};t\right).
\end{align*}
This completes the mathematical induction.
\end{IEEEproof}
\begin{rem}
For any $n<\infty$, the quantity $\frac{n}{n-1}\rho$ is the scaled average number of initial error bits attached to a code-constraint node, and $\frac{n}{n-1}\lambda_{n}^{(\ell)}$ is the scaled average number of error messages passed to a code-constraint node after the $\ell$-th iteration. Since $\frac{n}{n-1}\lambda_{n}^{(\ell)}\rightarrow\lambda^{(\ell)}$, it follows that the recursion (\ref{eq: PoiDE}) tracks the evolution of the average number of error messages passed to a code-constraint node.
\end{rem}

The following lemma shows that the DE recursion for the GLDPC ensemble whose component code is the even-weight subcode of a BCH code can be obtained by modifying (\ref{eq: PoiDE}). 
\begin{lem}
\label{lem: subcodes}Consider the GLDPC ensemble whose component code is the even-weight subcode of a BCH code. If $t$ is even, then the recursion for $\lambda^{(\ell)}$ is $\lambda^{(\ell+1)}=f_{e}\left(\lambda^{(\ell)};\rho\right)$ with
\begin{align*}
f_{e}\left(\lambda^{(\ell)};\rho\right) & \triangleq\rho\phi\left(\lambda^{(\ell)};t-1\right)+\frac{1}{(t-1)!}\psi\left(\lambda^{(\ell)};t\right).
\end{align*}
If $t$ is odd, the recursion is $\lambda^{(\ell+1)}=f_{o}\left(\lambda^{(\ell)};\rho\right)$ with 
\begin{align*}
f_{o}\left(\lambda^{(\ell)};\rho\right) & \triangleq\rho\phi\left(\lambda^{(\ell)};t-1\right)+\frac{1}{(t-1)!}\varphi\left(\lambda^{(\ell)};t\right).
\end{align*}
\end{lem}
\begin{IEEEproof}
See Appendix \ref{app: pf_subcodes}.
\end{IEEEproof}
For the spatially-coupled GLDPC ensemble,%
{} let $\lambda_{i}^{(\ell)}$ with $i\in\{1,2,\dots,L\}$ be the average number of error messages emitted by bit nodes at positions $i$ in the $\ell$-th iteration. We set $\lambda_{i}^{(0)}=\rho$ for all $i\in\{1,2,\dots,L\}$, and set $\lambda_{i}^{(\ell)}=0$ for all $i\notin\{1,2,\dots,L\}$ and $\ell\geq0$. Again, very similar arguments can be used to derive the DE for the spatially-coupled ensemble. Similar to (\ref{eq: SCDE3}), the resulting recursion is %
{} 
\begin{align}
\lambda_{i}^{(\ell+1)} & =\frac{1}{w}\sum_{k=0}^{w-1}f\left(\frac{1}{w}\sum_{j=0}^{w-1}\lambda_{i-j+k}^{(\ell)};\rho\right)\label{eq: SCPoiDE}
\end{align}
for $i\in\{1,2,\dots,L\}$. When the even-weight subcode of a BCH code is used as a component code in the spatially-coupled GLDPC ensemble, the recursion becomes
\begin{align*}
\lambda_{i}^{(\ell+1)} & =\begin{cases}
{\displaystyle \frac{1}{w}\sum_{k=0}^{w-1}f_{e}\left(\frac{1}{w}\sum_{j=0}^{w-1}\lambda_{i-j+k}^{(\ell)};\rho\right)}\vspace{1.5mm} & \mbox{ if }t\mbox{ is even,}\\
{\displaystyle \frac{1}{w}\sum_{k=0}^{w-1}f_{o}\left(\frac{1}{w}\sum_{j=0}^{w-1}\lambda_{i-j+k}^{(\ell)};\rho\right)} & \mbox{ if }t\mbox{ is odd.}
\end{cases}
\end{align*}

\vspace{0mm}

\section{Bounds on the Noise Threshold\label{sec:Bounds}}

In this section, we consider the noise thresholds of the decoding algorithms in Section \ref{sec: Ensembles} for the spatially-coupled ensemble when $L\gg w$ and $w\rightarrow\infty$. We employ the analysis proposed by Yedla \emph{et al.} \cite{Yedla-istc12,Yedla-it14} to compute these thresholds. Similarly, one could apply the results of~\cite{Kudekar-unpub13}.

Consider the recursion defined by the DE update equation of a decoding algorithm. Let $x$ be the average error probability of the messages emitted by bit nodes. For each channel parameter $p\in[0,1]$, the \emph{potential function}, denoted by $U(x;p)$, is proposed in \cite{Yedla-istc12,Yedla-it14}. If the recursion is a \emph{proper admissible system} (see \cite[Definition 31]{Yedla-it14}), then the \emph{coupled threshold} $p^{**}$ is defined as the supremum of the channel parameter $p$ such that the potential function is non-negative. The results in~\cite{Yedla-istc12,Yedla-it14} also show that the noise threshold of the decoding algorithm saturates to $p^{**}$ when the algorithm is applied to a spatially-coupled system with $L\gg w$ and $w\rightarrow\infty$. 

While the exact value of $p^{**}$ can be computed numerically for a given potential function, a simple expression is not available in terms of $t$ and $n$. Hence, we derive a suitable lower bound that is a simple function of $t$ and $n$.

\subsection{\label{subsec:IdealIterativeHDD} Iterative HDD with Ideal Component Decoders}

Suppose that one ignores the effect of miscorrection and considers the natural hard-decision peeling decoder for the $\left(\mathcal{C},m\right)$ ensemble based on BCH codes, then it is easy to see that at most $mt$ errors can be corrected using BDD. To achieve this upper bound, it must happen that each code corrects exactly $t$ errors. If some codes decode with fewer than $t$ errors, then there is an irreversible loss of error-correcting potential. Since there are $\frac{nm}{2}$ code bits, normalizing this number shows that the noise threshold is upper bounded by $\frac{2t}{n}$. In terms of the average number of errors in each code constraint, the threshold is upper bounded by $2t$ because each code involves $n$ bits. 

Before we delve into the analysis of the potential threshold for the iterative HDD with ideal component decoders, we first recall some definitions. The beta function, $B(a,b)$, and the normalized incomplete beta function, $I_{x}(a,b)$, \cite[\S 8.17]{Olver-2010} are given by 
\begin{align}
B(a,b) & \triangleq\int_{0}^{1}z^{a-1}(1-z)^{b-1}\mathrm{d}z\label{eq: Beta}
\end{align}
and
\begin{align}
I_{x}(a,b) & \triangleq\frac{1}{B(a,b)}\int_{0}^{x}z^{a-1}(1-z)^{b-1}\mathrm{d}z\triangleq\frac{B_{x}(a,b)}{B(a,b)}.\label{eq: IncompBeta}
\end{align}
When $a$ and $b$ are positive integers, these functions have the following well-known properties \cite[\S 8.17]{Olver-2010}: 
\begin{align}
B(a,b) & =\frac{(a-1)!(b-1)!}{(a+b-1)!},\nonumber \\
I_{x}(a,b) & =\sum_{k=a}^{a+b-1}\binom{a+b-1}{k}x^{k}(1-x)^{a+b-1-k},\label{eq: IandBinom}\\
B(a+1,b) & =\frac{a}{a+b}B(a,b),\label{eq: BandB}\\
I_{x}(a+1,b) & =I_{x}(a,b)-\frac{x^{a}(1-x)^{b}}{aB(a,b)}.\label{eq: IandI}
\end{align}
By differentiating both sides, one can also verify that
\begin{equation}
\int_{0}^{x}I_{z}(a,b)\mathrm{d}z=xI_{x}(a,b)-\frac{a}{a+b}I_{x}(a+1,b).\label{eq: intI}
\end{equation}

Since the DE update equation (\ref{eq: NomissDE}) for the iterative HDD without miscorrection is a scalar recursion. The following lemma enables us to apply the results in~\cite{Yedla-it14} to the current noise threshold analysis.
\begin{lem}
\label{lem: sas}%
Let the recursion (\ref{eq: NomissDE}) be represented as $x^{(\ell+1)}=\tf(\tg(x^{(\ell)});p)$ with $\tf(x;p)=px$ and $\tg(x)=\hfn(x)$. Then, for $t\geq2$ and $n\geq t+2$, the pair $(\tf,\tg)$ is a proper admissible system according to~\cite[Definition~31]{Yedla-it14} and unconditionally stable according to~\cite[Definition~34]{Yedla-it14}.
\end{lem}
\begin{IEEEproof}
It is easy to verify that $\tf(x;p)$ is strictly increasing in both $x$ and $p$. Using (\ref{eq: IandBinom}) and (\ref{eq: IncompBeta}), we can rewrite $\tg(x)$ as 
\begin{align}
\tg(x) & =\hfn(x)=\frac{1}{B(t,n-t)}\int_{0}^{x}z^{t-1}(1-z)^{n-t-1}\mathrm{d}z\label{eq: fn_eq_Ix}
\end{align}
and observe that $\tg(0)=0$. Since $\tg'(x)\triangleq\frac{\d\tg(x)}{\d x}=\frac{1}{B(t,n-t)}x^{t-1}(1-x)^{n-t-1}>0$ for all $x\in(0,1)$, $\tg(x)$ is strictly increasing in $x$. Also, the second derivative%

\begin{align*}
\frac{\d^{2}\tg(x)}{\d x^{2}} & =\frac{1}{B(t,n-t)}(t-1)x^{t-2}(1-x)^{n-t-1}\\
 & \qquad-\frac{1}{B(t,n-t)}(n-t-1)x^{t-1}(1-x)^{n-t-2}
\end{align*}

is a continuous function on $[0,1]$ for $t\geq2$ and $n\geq t+2$. The recursion is proper because $\frac{\d}{\d p}p\hfn(x)=\hfn(x)>0$ for $x\in(0,1]$. Finally, the recursion is unconditionally stable for $t\geq2$ because $\frac{\d}{\d x}p\hfn(x)\big|_{x=0}=p\hfn'(0)=0$. 
\end{IEEEproof}

Since the DE update (\ref{eq: NomissDE}) defines a proper admissible system for any integer $n>t$, the formulas for the potential function and the coupled threshold are given by~\cite[(16),Lemma~46(ii)]{Yedla-it14}.

\begin{defn}
\label{def: NoMissPotential}%
The \emph{potential function} associated with the recursion (\ref{eq: NomissDE}) is defined by 
\begin{align}
\hV_{n}(x;p) & \triangleq\int_{0}^{x}(z-p\hfn(z))\hfn'(z)\mathrm{d}z\nonumber \\
 & =x\hfn(x)-\int_{0}^{x}\hfn(z)\mathrm{d}z-\frac{1}{2}p\left(\hfn(x)\right)^{2}.\label{eq: potential_fun}
\end{align}
The \emph{coupled threshold} associated with (\ref{eq: SCDE2}) is defined to be
\begin{equation}
\hpn^{**}\triangleq\sup\left\{ p\in[0,1]\,\big|\,\min_{x\in[0,1]}\hV_{n}(x;p)\geq0\right\} .\label{eq: PotentialThreshold}
\end{equation}
This threshold for iterative HDD without miscorrection is achieved by $\left(\mathcal{C},m,L,w\right)$ spatially-coupled GLDPC ensembles in the limit where $m\gg L\gg w$ as $w\to\infty$. 
\end{defn}
\begin{rem}
\label{rem:SCproof} For $p<\hpn^{**}$, \cite[Lemma~46(ii)]{Yedla-it14}~implies that $\arg\min_{x\in[0,1]}\hV_{n}(x;p)=0$. Since the system is unconditionally stable, combining \cite[Proposition~10(iv)]{Yedla-it14} and \cite[Theorem~1]{Yedla-it14} implies that, for $p<\hpn^{**}$, there is a $w_{0}<\infty$ such that, for all $w>w_{0}$, the error probabilities $y_{i}^{(\ell)}$ for all $i\in\{1,2,\dots,L+w-1\}$ in~(\ref{eq: SCDE2}) converge to $0$ as $\ell\rightarrow\infty$. Furthermore, this implies the same result for the recursion~(\ref{eq: SCDE}) because $x_{i}^{(\ell)}\rightarrow0$ for all $i\in\{1,2,\dots,L\}$ if and only if $y_{i}^{(\ell)}\rightarrow0$ for all $i\in\{1,2,\dots,L+w-1\}$.
\end{rem}

\begin{rem}
From~(\ref{eq: potential_fun}), we see that $\hV_{n}(x;p)$ is decreasing in $p$. Therefore, one can obtain $\hpn^{**}$ by numerically computing 
\begin{align*}
\hpn^{**} & =\inf_{x\in(0,1]}\frac{x\hfn(x)-\int_{0}^{x}\hfn(z)\d z}{\frac{1}{2}\hfn^{2}(x)}.
\end{align*}
\end{rem}
\begin{rem}
Since the minimum in~\eqref{eq: PotentialThreshold} always occurs at a fixed-point of the recursion~\cite[Lemma~18]{Yedla-it14}, it is sufficient to consider the value of the potential at fixed points. For any $p\in(0,1]$, let $x$ be a non-zero fixed point of the recursion (\ref{eq: NomissDE}). Then, the fixed-point equation $x=p\hf_{n}(x)$ shows that any fixed-point $x$ defines an unique $p(x)\triangleq\frac{x}{\hfn(x)}$.
\end{rem}
\begin{defn}
\label{def: fixed-point potential} The \emph{fixed-point potential} for a proper admissible system is defined in~\cite[Def.~41]{Yedla-it14} to be 
\begin{align}
\hV_{n}(x) & \triangleq\hV_{n}(x;p(x)),\label{eq: fixed_point_pot}
\end{align}
though it is denoted there by $Q(x)$. This function represents the potential evaluated at the unique fixed-point defined by $x$.
\end{defn}
Combining Definitions~\ref{def: fixed-point potential} and~\ref{def: NoMissPotential}, we can compute
\begin{align}
\hV_{n}(x) & =x\hfn(x)-\int_{0}^{x}\hfn(z)\mathrm{d}z-\frac{1}{2}\frac{x}{\hfn(x)}\left(\hfn(x)\right)^{2}\nonumber \\
 & =\frac{1}{2}x\hfn(x)-\int_{0}^{x}\hfn(z)\mathrm{d}z\label{eq: fixed_point_pot_f}\\
 & =\frac{1}{2}xI_{x}(t,n\!-\!t)\!-\!xI_{x}(t,n\!-\!t)\!+\!\frac{t}{n}I_{x}(t\!+\!1,n\!-\!t)\nonumber \\
 & =\frac{-1}{2}xI_{x}(t,n\!-\!t)\!+\!\frac{t}{n}\!\left(\!I_{x}(t,n\!-\!t)\!-\!\frac{x^{t}(1\!-\!x)^{n-t}}{tB(t,n\!-\!t)}\right)\nonumber \\
 & =\frac{1}{2}\left(\frac{2t}{n}-x\right)I_{x}(t,n-t)-\frac{x^{t}(1-x)^{n-t}}{nB(t,n-t)},\label{eq: fixed_point_pot_I}
\end{align}
where the intermediate steps make use of~(\ref{eq: fn_eq_Ix}),~(\ref{eq: intI}), and~(\ref{eq: IandI}).
\begin{lem}
\label{lem: Vn_unique_root} For $t\geq2$ and $n\geq t+2$, the fixed-point potential $\hV_{n}(x)$ has a unique root, $\hat{x}_{n}^{**}$, satisfying $\hat{x}_{n}^{**}\in(0,1]$.
\end{lem}
\begin{IEEEproof}
See Appendix \ref{app: pf_Vnx_unique_root}.
\end{IEEEproof}
\begin{lem}
\label{lem: Vnxn=00003D0} Let $\hx_{n}^{**}$ be the unique value in $(0,1]$ satisfying $\hV_{n}(\hx_{n}^{**})=0$. Then, there exists a $t_{0}\geq3$ and a function $n_{0}(t)\geq t+2$ such that, for all $t\geq t_{0}$ and $n\geq n_{0}(t)$, we have $\frac{2t-2}{n}\leq\hat{x}_{n}^{**}\leq\frac{2t}{n}$ . 
\end{lem}
\begin{IEEEproof}
See Appendix \ref{app: pf_Vnx=00003D0}.
\end{IEEEproof}
\begin{rem}
The constant $t_{0}$ and function $n_{0}(t)$ are required for our rigorous proof that $\hV_{n}(\frac{2t-2}{n})>0$ when $t\geq t_{0}$ and $n\geq n_{0}(t)$. However, through the numerical evaluation, one can observe that $\hV_{n}(\frac{2t-2}{n})>0$ for all $t\geq2$ and $n>2t$. Thus, we conjecture that the lemma holds for this larger range.
\end{rem}
\begin{thm}
For $t\geq2$ and $n\geq t+2$, the coupled threshold associated with (\ref{eq: SCDE2}) is given by $\hpn^{**}=p(\hat{x}_{n}^{**})=\hat{x}_{n}^{**}/\hfn(\hat{x}_{n}^{**})$, where $\hat{x}_{n}^{**}$ is the unique root of $\hV_{n}(x)=0$ satisfying $\hat{x}_{n}^{**}\in(0,1]$. Moreover, there exists a $t_{0}\geq3$ and a function $n_{0}(t)\geq t+2$ such that, for all $t\geq t_{0}$ and $n\geq n_{0}(t)$, we have
\[
\hpn^{**}\geq\hat{x}_{n}^{**}\geq\frac{2t-2}{n}.
\]
 
\end{thm}
\begin{IEEEproof}
For $t\geq2$ and $n\geq t+2$, Lemma~\ref{lem: Vn_unique_root} shows that $\hV_{n}(x)=0$ has a unique root, $\hat{x}_{n}^{**}$, satisfying $\hat{x}_{n}^{**}\in(0,1]$. Since Lemma~\ref{lem: sas} shows that (\ref{eq: NomissDE}) is a proper admissible system that is unconditionally stable, we can apply \cite[Lemma~46(iii)]{Yedla-it14} to see that the coupled threshold satisfies $\hpn^{**}=p(\hat{x}_{n}^{**})$. The uniqueness of the root allows one to avoid computing the minimum in \cite[Lemma~46(iii)]{Yedla-it14}. From the definition of $p(x)$, we observe that
\[
p(\hat{x}_{n}^{**})=\frac{\hat{x}_{n}^{**}}{\hfn(\hat{x}_{n}^{**})}\geq\hat{x}_{n}^{**}.
\]
Finally, Lemma~\ref{lem: Vnxn=00003D0} shows that there exists a $t_{0}\geq3$ and a function $n_{0}(t)\geq t+2$ such that, for all $t\geq t_{0}$ and $n\geq n_{0}(t)$, we have $\hat{x}_{n}^{**}\geq\frac{2t-2}{n}$.
\end{IEEEproof}

\begin{rem}
 One may also view the update equation (\ref{eq: NomissDE}) as the recursion for the combined function pair $(p\hfn(x),x)$. Then, the DE update equation for the corresponding spatially-coupled system is 
\begin{align*}
y_{i}^{(\ell+1)} & =\frac{1}{w}\sum_{j=\max\{i-L,0\}}^{\min\{i-1,w-1\}}p\hf_{n}\left(\frac{1}{w}\sum_{k=0}^{w-1}y_{i-j+k}^{(\ell)}\right),
\end{align*}
where $y_{i}^{(\ell)}$ is the average error probability of the input messages to the $\tg$-nodes at the $i$-th position and $i\in\{1,2,\dots,L+w-1\}$. Using the proof of Lemma~\ref{lem: sas}, one can show that the combined system $(p\hfn(x),x)$ is also a proper admissible system. By the definition of the potential function in~\cite{Yedla-it14}, the potential function for the combined ideal system, $(p\hfn(x),x)$, is $\hU_{n}(x;p)\triangleq\int_{0}^{x}(z-p\hfn(z))\mathrm{d}z$. This is the potential function used in \cite{Jian-isit12}, but it is denoted there by $U_{n}(x;p)$. Let $\hU_{n}(x)\triangleq\hU_{n}(x;p(x))$ be the fixed-point potential of the $(p\hfn(x),x)$ system.  Since the fixed points of the $(px,\hfn(x))$ and $(p\hfn(x),x)$ systems are very closely related, one can show that their fixed-point potential functions must satisfy
\begin{align*}
\hV_{n}(x) & =\frac{2\hfn(x)}{x}\hU_{n}(x)
\end{align*}
for all $x\in(0,1]$ satisfying $x=p\hfn(x)$. This is similar to the half-iteration shift described in~\cite[\S II-D]{Yedla-it14} 
\end{rem}

For the recursion defined by the high-rate scaling limit, (\ref{eq: HighRateNoMiss}), it is easy to show that $\hf(\lambda;\rho)=\rho\phi(\lambda;t-1)$ is increasing in both $\lambda$ and $\rho$. Thus, the noise threshold for the recursion (\ref{eq: HighRateNoMiss}), in terms of the average number of errors per $n-1$ bits, exists and is denoted by $\hr_{t}^{*}$ \cite[\S 3.10 -- \S 3.11]{RU-2008}. Also, by the monotonicity property of $\hf(\lambda;\rho)$, the noise threshold for the recursion (\ref{eq: SCPoiDENoMiss}) exists as well. Using \cite{Yedla-istc12,Yedla-it14}, we define the potential function and the coupled threshold for the recursion (\ref{eq: HighRateNoMiss}), respectively, by
\begin{align*}
\hV(\lambda;\rho) & \triangleq\int_{0}^{\lambda}\left(z-\hat{f}\left(z;\rho\right)\right)\hat{f}'\left(z;\rho\right)\d z\\
 & =\int_{0}^{\lambda}\left(z-\rho\phi(z;t-1)\right)\phi'(z;t-1)\d z,
\end{align*}
and 
\begin{align}
\hr_{t}^{**} & \triangleq\sup\left\{ \rho\in[0,\infty)\,\big|\,\min_{\lambda\geq0}\hV\left(\lambda;\rho\right)\geq0\right\} .\label{eq: HighRatePotentialThreshold}
\end{align}
Using the same argument as in Remark~\ref{rem:SCproof}, one finds that the coupled threshold $\hr_{t}^{**}$ can also be achieved by applying to the $(\Cc,m,L,w)$ spatially-coupled ensemble in the limit where $m\gg L\gg n\gg w$ as $w\rightarrow\infty$.
\begin{cor}
\label{cor: high-rate_thresh_no_miscorr} For the high-rate scaling limit of the recursion (\ref{eq: HighRateNoMiss}), the coupled threshold (in terms of the average number of errors in a code constraint) satisfies $\hat{\rho}_{t}^{**}\geq2t-2$ for all $t\geq t_{0}$.
\end{cor}
\begin{IEEEproof}
Consider the DE recursion (\ref{eq: NomissDE}) for the $(\Cc,m)$ GLDPC ensemble. Let $\hpn^{**}$ be the coupled threshold as defined in (\ref{eq: PotentialThreshold}). Note that $\hpn^{**}$ can be achieved by the $(\Cc,m,L,w)$ spatially-coupled ensemble in the limit where $m\gg L\gg n\gg w$ as $w\rightarrow\infty$. For any fixed $t\geq t_{0}$ and $n\geq n_{0}(t)$, we define $\hat{\rho}_{n,t}^{**}\triangleq(n-1)\hpn^{**}$. From Lemma \ref{lem: Vnxn=00003D0}, we know that $\hat{\rho}_{n,t}^{**}\geq(2t-2)\frac{n-1}{n}$ for all $t\geq t_{0}$ and $n\geq n_{0}(t)$. Thus, we conclude that $\hat{\rho}_{t}^{**}\triangleq\lim_{n\rightarrow\infty}\hr_{n,t}^{**}\geq2t-2$ for all $t\geq t_{0}$.
\end{IEEEproof}

\subsection{Iterative HDD with BDD}

For the case of iterative HDD algorithm, it is not clear if the recursion (\ref{eq: fnDE}) defines an admissible system for every linear code (e.g., monotonicity could fail). However, we believe that is the case. Fortunately, one can show that in the high-rate scaling limit, the function $f(\lambda;\rho)$ in (\ref{eq: PoiDE}) is strictly increasing in both arguments for $\lambda,\rho>0$. Therefore, the noise threshold and the coupled threshold for the recursion (\ref{eq: PoiDE}), denoted by $\rho^{*}$ and $\rho^{**}$ respectively, exist as well. Using the setup in \cite{Yedla-it14} and the same argument as in Remark~\ref{rem:SCproof}, one finds that the potential function and the coupled threshold for the recursion (\ref{eq: PoiDE}) are, respectively, 
\begin{align*}
U(\lambda;\rho) & \triangleq\int_{0}^{\lambda}\left(z-f\left(z;\rho\right)\right)f'\left(z;\rho\right)\d z,
\end{align*}
and 
\begin{align*}
\rho_{t}^{**} & \triangleq\sup\left\{ \rho\in[0,\infty)\,\big|\,\min_{\lambda\geq0}U\left(\lambda;\rho\right)\geq0\right\} .
\end{align*}
We note that this potential function is different from the $U(\lambda;\rho)$ function defined in~\cite{Jian-isit12}, which is based on the alternative scalar system.

Using the threshold of iterative HDD without miscorrection, $\hat{\rho}_{t}^{**}$, one can obtain a lower bound of $\rho_{t}^{**}$ from the following lemma.
\begin{lem}
\label{lem: high-rate_thresh_w_miscorr} For the high-rate scaling limit, the coupled threshold of iterative HDD with miscorrection, $\rho_{t}^{**}$, satisfies $\rho_{t}^{**}\geq\hat{\rho}_{t}^{**}-\frac{1}{(t-1)!}$.
\end{lem}
\begin{IEEEproof}
To show the lower bound of $\rho_{t}^{**}$, we introduce the recursion
\begin{align}
\ovl^{(\ell+1)} & =\left(\rho+\frac{1}{(t-1)!}\right)\phi\left(\ovl^{(\ell)};t-1\right).\label{eq: HighRateB}
\end{align}
Let $\ovr\triangleq\rho+\frac{1}{(t-1)!}$. From (\ref{eq: PoiDE}), one can show that 
\begin{align}
f\left(\lambda;\rho\right) & \leq\rho\phi(\lambda;t-1)+\frac{1}{(t-1)!}\phi(\lambda;t-1)\nonumber \\
 & =\ovr\phi(\lambda;t-1),\label{eq: rhobarRec}
\end{align}
and thus 
\begin{align*}
\ovl^{(\ell+1)} & \geq f\left(\ovl^{(\ell)};\rho\right).
\end{align*}
Therefore, we know that $\lambda^{(\ell)}$ in the recursion (\ref{eq: PoiDE}) satisfies $\lambda^{(\ell)}\leq\ovl^{(\ell)}$ when the initial values $\lambda^{(0)}=\ovl^{(0)}$. 

By rewriting (\ref{eq: HighRateB}) as $\ovl^{(\ell+1)}=\ovr\phi(\ovl^{(\ell)};t-1)$, the recursion (\ref{eq: HighRateB}) can be considered as a $\left(\ovr x,\phi(x;t-1)\right)$ system. The update equation of the spatially-coupled $\left(\ovr x,\phi(x;t-1)\right)$ system is
\begin{align}
\ovl_{i}^{(\ell+1)} & =\ovr\left(\frac{1}{w}\sum_{k=0}^{w-1}\phi\left(\frac{1}{w}\sum_{j=0}^{w-1}\ovl_{i-j+k}^{(\ell)};t-1\right)\right),\label{eq: SC_rhobar}
\end{align}
where $\ovl_{i}^{(\ell)}$ is the average number of error messages emitted by bit nodes at position $i$ in the $\ell$-th iteration and $i\in\{1,2,\dots,L\}$. From (\ref{eq: rhobarRec}), we know 
\begin{align*}
\ovl_{i}^{(\ell+1)} & =\frac{1}{w}\sum_{k=0}^{w-1}\ovr\phi\left(\frac{1}{w}\sum_{j=0}^{w-1}\ovl_{i-j+k}^{(\ell)};t-1\right)\\
 & \geq\frac{1}{w}\sum_{k=0}^{w-1}f\left(\frac{1}{w}\sum_{j=0}^{w-1}\ovl_{i-j+k}^{(\ell)};\rho\right).
\end{align*}
With the same $w$ and the same initial value, $\lambda_{i}^{(0)}=\ovl_{i}^{(0)}$ for $i\in\{1,2,\dots,L\}$, $\lambda_{i}^{(\ell)}$ in (\ref{eq: SCPoiDE}) is upper bounded by $\ovl_{i}^{(\ell)}$ for $i\in\{1,2,\dots,L\}$ and $\ell\geq0$.

Denote the coupled threshold for the recursion (\ref{eq: rhobarRec}) by $\ovr^{**}$. We know that $\ovr^{**}=\hr^{**}$, where $\hr^{**}$ is defined in (\ref{eq: HighRatePotentialThreshold}). Using \cite{Yedla-it14} and the argument in Remark~\ref{rem:SCproof}, we know that, for each $\rho<\ovr^{**}-\frac{1}{(t-1)!}=\hat{\rho}_{t}^{**}-\frac{1}{(t-1)!}$, there exists a $w_{\rho}>0$ such that all $w>w_{\rho}$, $\ovl_{i}^{(\ell)}\rightarrow0$ as $\ell\rightarrow\infty$ for all $i\in\{1,2,\dots,L\}$. By the fact that $\lambda_{i}^{(\ell)}\leq\ovl_{i}^{(\ell)}$ for all $i\in\{1,2,\dots,L\}$, we know that $\lambda_{i}^{(\ell)}\rightarrow0$ as well. This implies that $\rho_{t}^{**}\geq\hat{\rho}_{t}^{**}-\frac{1}{(t-1)!}$.
\end{IEEEproof}

\section{Approaching Capacity \label{sec: Approaching}}

In this section, we show that the proposed iterative HDD for the spatially-coupled ensemble can approach the capacity in the high-rate regime. 

\subsection{A Sequence of Ensembles with Vanishing Redundancy}

Consider a sequence of ensembles, for $\nu=1,2,\ldots$, with component-code lengths $n_{\nu}\triangleq2^{\nu}-1$, rates $R_{\nu}=1-\frac{2\nu t}{n_{\nu}}$ and noise thresholds $p_{\nu}^{*}=\frac{2t}{n_{\nu}}$. The following lemma shows that, for any $\epsilon>0$, the ensembles in this sequence are eventually $\epsilon$-redundancy achieving on the BSC. That is, for any $\epsilon>0$, there exists a $\nu_{0}\in\mathbb{N}$ such that, for all $\nu\geq\nu_{0}$, one has 
\begin{align*}
\frac{1-C(2tn_{\nu}^{-1})}{2t\nu n_{\nu}^{-1}} & \geq1-\epsilon.
\end{align*}

\begin{lem}
\label{lem:HighRateCap} Consider a sequence of BSCs with error probability $2tn_{\nu}^{-1}$ for a fixed $t$ and $\nu\in\mathbb{N}$. Then, the ratio of $1-C\left(2tn_{\nu}^{-1}\right)$ to $2t\nu n_{\nu}^{-1}$ goes to $1$ as $\nu\rightarrow\infty$. That is,\vspace{0mm} 
\begin{align}
\lim_{\nu\rightarrow\infty}\frac{1-C\left(2tn_{\nu}^{-1}\right)}{2t\nu n_{\nu}^{-1}} & =1.\label{eq: CapacityRatioLimit}
\end{align}
\end{lem}
\begin{IEEEproof}
Recall that $C(p)=1-H(p)$, where $H(p)=-p\log_{2}(p)-(1-p)\log_{2}(1-p)$ is the binary entropy function. The numerator of the LHS of (\ref{eq: CapacityRatioLimit}) can be written as 
\begin{equation}
H\left(\frac{2t}{n_{\nu}}\right)=\frac{2t\log_{2}n_{\nu}}{n_{\nu}}\left(1-\frac{\log_{2}\left(\frac{2t}{e}\right)}{\log_{2}n_{\nu}}-O\left(n_{\nu}^{-1}\right)\right)\!.\label{eq: Entropy}
\end{equation}
By substituting (\ref{eq: Entropy}) into the LHS of (\ref{eq: CapacityRatioLimit}), we have
\begin{align*}
\frac{1-C\left(2tn_{\nu}^{-1}\right)}{2t\nu n_{\nu}^{-1}} & =\frac{2tn_{\nu}^{-1}\log_{2}\left(n_{\nu}\right)}{2t\nu n_{\nu}^{-1}}\left(1-O\left(\nu^{-1}\right)\right).
\end{align*}
Then, the limit in~(\ref{eq: CapacityRatioLimit}) follows because $\log_{2}(n_{\nu})=\!\nu\!+\!o(1).$ 
\end{IEEEproof}

\subsection{Noise Threshold of Iterative HDD when $n<\infty$}

In the following discussion, we use $(n_{\nu},k,2t+1)$ binary primitive BCH codes as component codes, where $n_{\nu}=2^{\nu}-1$. Let the noise threshold of iterative HDD with BDD for the $(\Cc,m,L,w)$ spatially-coupled GLDPC ensemble, denoted by $p_{n_{\nu}}^{**}$, be as defined in (\ref{eq: SCBDDthreshold}). We first show that $p_{n_{\nu}}^{**}$ satisfies $(n_{\nu}-1)p_{n_{\nu}}^{**}\geq2t-2-\frac{1}{(t-1)!}-\epsilon$ for some $t>0$, $\epsilon>0$ and $L\gg w>0$. Then, we show that the $(\Cc,m,L,w)$ spatially-coupled GLDPC code is $\epsilon$-redundancy achieving when $n\gg t\gg1$ in Theorem \ref{thm: main}.
\begin{lem}
\label{lem: EnQnFinite_n}For any $0\leq\lambda<\infty$ and $0<t<\infty$, let $n$ satisfy $\lfloor\sqrt{n}\rfloor>\max\{e\lambda,\lambda+1+t\}$, and let $X_{n}\sim\mbox{Bi}(n-1,\frac{\lambda}{n-1})$ be a binomial random variable with the mean $\lambda$. Then, %
{} 
\begin{align*}
E\left[nQ_{n}(X_{n})\right] & \!\leq\!\left(\!\frac{1}{(t\!-\!1)!}\!+\!O\left(n^{-0.1}\right)\!\right)I_{\frac{\lambda}{n-1}}(t\!+\!1,n\!-\!t\!-\!1).
\end{align*}
\end{lem}
\begin{IEEEproof}
See Appendix \ref{app: pf_Finite_n_Rec}.
\end{IEEEproof}

\begin{lem}
\label{lem: Finite_n_RecUB}Given a pair $(t,\rho)$ with $t>0$ and $\rho\geq2$, and for each $\epsilon>0$, there exists a $n_{1}>0$ such that, for all $n\geq n_{1}$, 
\begin{align*}
(n\!-\!1)f_{n}\!\left(\frac{\lambda}{n\!-\!1};\frac{\rho}{n\!-\!1}\right)\! & \leq\!\left(\!\rho\!+\!\frac{1}{(t\!-\!1)!}\!+\!\epsilon\!\right)\phi\left(\lambda;t\!-\!1\right),
\end{align*}
for all $0\leq\lambda\leq3\rho$, where $f_{n}(x;p)$ is defined in (\ref{eq: fn}).
\end{lem}
\begin{IEEEproof}
From (\ref{eq: ScaleDE2}), we know 
\begin{align*}
 & (n-1)f_{n}\left(\frac{\lambda}{n-1};\frac{\rho}{n-1}\right)\\
 & \quad\leq\rho I_{\frac{\lambda}{n-1}}\left(t,n-t\right)\\
 & \qquad\quad+\sum_{i=t+1}^{n-1}\!\!\binom{n\!-\!1}{i}\!\!\left(\frac{\lambda}{n\!-\!1}\right)^{i}\!\!\left(\!1\!-\!\frac{\lambda}{n\!-\!1}\right)^{n-1-i}\!nQ_{n}(i)\\
 & \quad\overset{\mbox{(a)}}{\leq}\rho I_{\frac{\lambda}{n-1}}\left(t,n-t\right)\\
 & \qquad\quad+\left(\frac{1}{(t-1)!}+O\left(n^{-0.1}\right)\right)I_{\frac{\lambda}{n-1}}(t+1,n-t-1)\\
 & \quad\leq\left(\rho+\frac{1}{(t-1)!}+O\left(n^{-0.1}\right)\right)I_{\frac{\lambda}{n-1}}\left(t,n-t\right).
\end{align*}
where the inequality (a) is obtained by applying Lemma \ref{lem: EnQnFinite_n}. From \cite{Hwang-tpa11}, we know that $I_{\frac{\lambda}{n}}(t,n-t+1)$ converges to $\phi(t-1,\lambda)$ uniformly for $\lambda\in[0,3\rho]$. For any $\epsilon>0$, there exists a $n_{1}>0$ such that $I_{\frac{\lambda}{n-1}}\left(t,n-t\right)\leq\phi(\lambda;t-1)+\frac{\epsilon}{2}(\rho+\frac{1}{(t-1)!}+\epsilon)^{-1}\phi(\lambda;t-1)$ for all $0\leq\lambda\leq3\rho$, and $O(n^{-0.1})\leq\frac{\epsilon}{2}$. Then, whenever $n\geq n_{1}$, 
\begin{align*}
 & (n-1)f_{n}\left(\frac{\lambda}{n-1};\frac{\rho}{n-1}\right)\\
 & \quad\leq\left(\rho+\frac{1}{(t-1)!}+\frac{\epsilon}{2}\right)\\
 & \qquad\quad\times\left(1+\frac{\epsilon}{2}\left(\rho+\frac{1}{(t-1)!}+\epsilon\right)^{-1}\right)\phi(\lambda;t-1)\\
 & \quad\leq\left(\rho+\frac{1}{(t-1)!}+\epsilon\right)\phi(\lambda;t-1).
\end{align*}
\end{IEEEproof}
For the spatially-coupled $(\Cc,m,L,w)$ GLDPC codes with iterative BDD, let $\x^{(\ell)}(\x^{(0)};p)\in[0,1]^{L}$ be the vector of error probabilities after $\ell$ iterations of (\ref{eq: SCDE3}) with the initial error probability vector $\x^{(0)}\in[0,1]^{L}$, where the $i$-th element, denoted by $x_{i}^{(\ell)}(\x^{(0)};p)$, is the average error probability of the messages emitted by bit nodes at the $i$-th position. We define the noise threshold for the recursion (\ref{eq: SCDE3}) by
\begin{align}
p_{n}^{**} & \triangleq\sup\left\{ \!p\!\in\!(0,\!1]\,\big|\,\lim_{\ell\rightarrow0}\x^{(\ell)}(p'\boldsymbol{1};p')\!=\!\boldsymbol{0},\,p'\!\in\![0,p]\!\right\} \!.\!\label{eq: SCBDDthreshold}
\end{align}

\begin{cor}
\label{cor: Finite_n_Threshold}Given a pair $(t,\rho)$ with $t>0$ and $\rho\geq2$, and for each $\epsilon>0$, there exists a $n_{1}>0$ such that, whenever $n_{1}\leq n<\infty$, the noise threshold $p_{n}^{**}$ satisfies $(n-1)p_{n}^{**}\geq2t-2-\frac{1}{(t-1)!}-\epsilon$ when $m\gg L\gg w$ and $w\rightarrow\infty$.
\end{cor}
\begin{IEEEproof}
According to Lemma \ref{lem: Finite_n_RecUB}, we consider the following recursion 
\begin{equation}
\ovl^{(\ell+1)}=\ovr\phi(\ovl^{(\ell)};t-1),\label{eq: Finite_n_RecUB}
\end{equation}
where 
\begin{equation}
\ovr\triangleq(n-1)p+\frac{1}{(t-1)!}+\epsilon.\label{eq: ovrDef}
\end{equation}

Since the recursion (\ref{eq: Finite_n_RecUB}) can also be considered as a $\left(\ovr x,\phi(x;t-1)\right)$ admissible system, and the DE update equation of the spatially-coupled system is shown in (\ref{eq: SC_rhobar}). For each $\epsilon>0$, let $n_{1}$ be selected according to the proof of Lemma \ref{lem: Finite_n_RecUB}. Then, we know that, when $n>n_{1}$,
\begin{align}
 & \frac{\ovl_{i}^{(\ell+1)}}{n-1}=\frac{1}{n-1}\ovr\left(\frac{1}{w}\sum_{k=0}^{w-1}\phi\left(\frac{1}{w}\sum_{j=0}^{w-1}\ovl_{i-j+k}^{(\ell)};t-1\right)\right)\nonumber \\
 & \quad\geq\frac{1}{n-1}\left(\frac{1}{w}\sum_{k=0}^{w-1}(n-1)f_{n}\left(\frac{1}{w}\sum_{j=0}^{w-1}\frac{\ovl_{i-j+k}^{(\ell)}}{n-1};\frac{\rho}{n-1}\right)\right)\nonumber \\
 & \quad=\frac{1}{w}\sum_{k=0}^{w-1}f_{n}\left(\frac{1}{w}\sum_{j=0}^{w-1}\frac{\ovl_{i-j+k}^{(\ell)}}{n-1};\frac{\rho}{n-1}\right).\label{eq: SCUB}
\end{align}
Note that the RHS of the last equality in (\ref{eq: SCUB}) is the update equation (\ref{eq: SCDE3}) with $x_{i}^{(\ell)}=\frac{\ovl_{i}^{(\ell)}}{n-1}$ and $p=\frac{\rho}{n-1}$. When $(n-1)x_{i}^{(0)}=\ovl_{i}^{(0)}=\lambda$ for all $i\in\{1,2,\dots,L\}$ and $p=\frac{\rho}{n-1}$, one can show that $\frac{\ovl_{i}^{(\ell+1)}}{n-1}\geq x_{i}^{(\ell)}$ for all $\ell\geq0$ by induction. 

Let $\ovr^{**}$ be the coupled threshold of the recursion (\ref{eq: Finite_n_RecUB}). For each $\ovr<\ovr^{**}$, there exists a $w_{\ovr}>0$ such that, for all $w>w_{\ovr}$, we have $\ovl_{i}^{(\ell)}\rightarrow0$ for all $i\in\{1,2,\dots,L\}$. Since $\frac{\ovl_{i}^{(\ell+1)}}{n-1}\geq x_{i}^{(\ell)}$, we also know that $x_{i}^{(\ell)}\rightarrow0$ whenever $p$ satisfies 
\begin{align*}
(n-1)p+\frac{1}{(t-1)!}+\epsilon & <\ovr^{**}.
\end{align*}
Thus, the coupled threshold $p_{n}^{**}$ is lower bounded by 
\begin{align*}
p_{n}^{**} & \geq\frac{1}{n-1}\left(\ovr^{**}-\frac{1}{(t-1)!}-\epsilon\right).
\end{align*}
Since $\ovr^{**}=\hr^{**}\geq2t-2$, we conclude that $(n-1)p_{n}^{**}\geq2t-2-\frac{1}{(t-1)!}-\epsilon$.
\end{IEEEproof}

\subsection{Proof of the Main Theorem\label{subsec: MainProof}}

Now, we present the proof of the main theorem.
\begin{IEEEproof}[Proof of Theorem \ref{thm: main}.]
We prove the theorem by showing the existence of a tuple $(t,n,L,w)$ such that, for a given $\epsilon\in(0,1)$, the $(\set C,\infty,L,w)$ GLDPC spatially-coupled ensemble with the proposed iterative HDD algorithm is $\epsilon$-redundancy achieving. 

First, let $t\geq\max\{\frac{8}{\epsilon},t_{0}\}$, where $t_{0}$ is defined in Lemma \ref{lem: Vnxn=00003D0}. Then, select a $\nu_{1}>0$ and define $n_{\nu_{1}}=2^{\nu_{1}}-1$ such that $H(\frac{2t}{n_{\nu_{1}}})\geq2t\nu n_{\nu_{1}}^{-1}(1-\frac{\epsilon}{4})\log_{2}n_{\nu_{1}}$ and $n_{\nu_{1}}\geq n_{0}(t)$. From the threshold of the high-rate scaling limit introduced in Lemma \ref{lem: high-rate_thresh_w_miscorr}, we know that the noise threshold of the spatially-coupled system is around $2t$. Thus, we can consider the channel noise in terms of the average number of errors per code $\rho\in[0,3t]$. By Lemma \ref{lem: Finite_n_RecUB}, there exists a $n_{1}>0$ such that $(n-1)f_{n}\left(\frac{\lambda}{n-1};\frac{\rho}{n-1}\right)\leq\left(\rho+\frac{1}{(t-1)!}+\frac{1}{2}\right)\phi\left(\lambda;t-1\right)$ for all $0\leq\lambda\leq3\rho$. Let $\nu\geq\lceil\max\{\log_{2}n_{t},\log_{2}n_{0},\log_{2}n_{1},\nu_{1}\}\rceil$ and $n=2^{\nu}-1$. From Corollary \ref{cor: Finite_n_Threshold}, we know the noise threshold of the spatially-coupled recursion of (\ref{eq: ScaleDE2}) satisfies $(n-1)p^{**}\geq2t-2-\frac{1}{(t-1)!}-\frac{1}{2}$. By selecting $p=(2t-4)(n-1)^{-1}$, we know that there exists a $0<w<\infty$ such that the spatially-coupled recursion of (\ref{eq: ScaleDE2}) converges to $0$ as $\ell\rightarrow\infty$. After determining $w$, we select $L$ such that $L\geq2(w-1)\epsilon^{-1}$. From (\ref{eq: SCR}), the rate of the spatially-coupled code satisfies $R\geq1-\frac{2t\nu}{n_{\nu}}(1+\frac{\epsilon}{2})$. Finally, we conclude the proof by showing that the $(\set C,m,L,w)$ spatially-coupled ensemble is $\epsilon$-redundancy achieving because 
\begin{align*}
\frac{1-C\left(p\right)}{1-R} & \geq\frac{\left(2t-4\right)\nu n_{\nu}^{-1}\left(1-\frac{\epsilon}{4}\right)}{2t\nu n_{\nu}^{-1}\left(1+\frac{\epsilon}{2}\right)}\\
 & =\frac{\left(2t-4\right)\left(1-\frac{\epsilon}{4}\right)}{2t\left(1+\frac{\epsilon}{2}\right)}\geq\frac{\left(1-\frac{\epsilon}{4}\right)^{2}}{\left(1+\frac{\epsilon}{2}\right)}\geq1-\epsilon.
\end{align*}
\end{IEEEproof}
\vspace{0mm}

\section{Practical Implementation of Iterative HDD\label{sec: LowComplex}}

In this section, we describe the practical implementation of the iterative HDD described in Section \ref{sec: Ensembles}. We highlight the difference between conventional decoding, which we call intrinsic message passing (IMP), and the proposed approach in Section \ref{sec: Ensembles}, which we call extrinsic message passing (EMP). In EMP algorithms, messages passed on edges in the Tanner graph are computed only from their extrinsic information. For certain random ensembles, this enables analysis via density evolution. We emphasize that this is different than the conventional iterative HDD rule typically used by product codes. In contrast to Section \ref{subsec: iterBDD}, this section introduces the EMP algorithm in a message-passing fashion to make it clear that the EMP uses only the extrinsic information.

Let $\boldsymbol{r}$ be the vector of channel output bits, $\nu_{i,j}^{(\ell)}\in\{0,1\}$ be the messages passed from the $i$-th bit node to the $j$-th constraint node in the $\ell$-th iteration, and $\mu_{i,j}^{(\ell)}\in\{0,1\}$ be the messages passed from the $j$-th constraint node to the $i$-th bit node in the $\ell$-th iteration. We assume the constraint nodes define an $(n,k,d_{\min})$ component code $\mathcal{C}$ with $d_{\min}=2t+1$. Let $\vect v\in\{0,1\}^{n}$ be an input vector to a constraint node, and let $\mathsf{D}:\{0,1\}^{n}\rightarrow\mathcal{C}\cup\{\texttt{fail}\}$ be the operator of bounded distance decoding (BDD) with decoding radius $t$ defined by 
\begin{align*}
\Dec(\vect v) & =\begin{cases}
\vect c & \mbox{ if }\vect c\in\set C\ \mbox{and}\ d_{H}(\vect v,\vect c)\leq t\\
\texttt{fail} & \mbox{ otherwise.}
\end{cases}
\end{align*}

\subsection{Intrinsic Message Passing }

In this section, we recall the IMP algorithm to highlight the difference with EMP. For a bit node $i$ and a constraint node $j$, let $\nbr(i)=\{j,j'\}$ be the constraint-node neighbors of $i$. Let $\boldsymbol{\nu}_{j}^{(\ell)}\triangleq(\nu_{\sigma_{j}(1),j}^{(\ell)},\dots,\nu_{\sigma_{j}(n),j}^{(\ell)})$ be the collection of the all input messages to the $j$-th constraint node in the $\ell$-th iteration, where $\sigma_{j}(k)\in\set I$ is defined in Section \ref{sec: Ensembles}. Let $\boldsymbol{w}_{j}^{(\ell)}\triangleq(w_{1,j}^{(\ell)},\dots,w_{n,j}^{(\ell)})=\mathsf{D}(\boldsymbol{\nu}_{j}^{(\ell)})$ be the output of the BDD decoder applied to the input messages at the $j$-th constraint node. Then, the message-passing rules, for each constraint node $j$, are
\begin{align}
\nu_{i,j}^{(\ell+1)} & =\mu_{i,j'}^{(\ell)}\label{eq: IterHDD1}\\
\mu_{\sigma_{j}(k),j}^{(\ell)} & =\begin{cases}
w_{k,j}^{(\ell)} & \mbox{ if }\mathsf{D}\big(\boldsymbol{\nu}_{j}^{(\ell)}\big)\neq\texttt{fail}\vspace{1ex}\\
\nu_{\sigma_{j}(k),j}^{(\ell)} & \mbox{ otherwise.}
\end{cases}\label{eq: IterHDD2}
\end{align}
The iteration starts by initializing $\nu_{i,j}^{(0)}=r_{i}$ for each bit node $i$ and all $j\in N(i)$. From (\ref{eq: IterHDD1}) and (\ref{eq: IterHDD2}), one can see that the IMP algorithm only uses channel outputs at the beginning of the iterations, and then, exchanges the output of the constraint nodes in the subsequent iterations. The IMP is the conventional approach used for the iterative HDD of product codes.

For a standard product code, the IMP decoder is essentially identical to the conventional decoder. A key point is that there are actually two distinct conventional decoders: one that starts by decoding the rows and one that starts by decoding the columns. The IMP decoder computes both of these answers simultaneously. Due to the Tanner graph structure, the messages associated with decoding the rows first are independent of the messages associated with decoding the columns first. Thus, while there are indeed two messages for each variable that may not be equal, they do not interact and are associated with these two possible first steps. Structurally, this happens because Tanner graph of a product code is actually tripartite (i.e., row and column code constraints are separated by a layer of variable nodes). 

\begin{table*}[tbh]
\caption{\vspace{0.5mm} The comparison between the proposed EMP implementation and the iterative HDD introduced in Section \ref{subsec: iterBDD}.\label{tab:EquivalentTable}}

\centering{}%
\begin{tabular}{|c|c|c|c|c|c|}
\hline 
$\left|\mathcal{K}\right|$ & $k$ & $r_{\sigma_{j}(k)}$ & $\nu_{\sigma_{j}(k),j'}^{(\ell+1)}$ & $\hat{\mu}_{\sigma_{j}(k),j}^{(\ell)}$ & $\hat{\nu}_{\sigma_{j}(k),j'}^{(\ell+1)}$\tabularnewline
\hline 
\hline 
$|\mathcal{K}|<t$ & $k\in\{1,2,\dots,n\}$ & 0 & 0 & $(0,0)$ & 0\tabularnewline
\hline 
$|\mathcal{K}|<t$ & $k\in\{1,2,\dots,n\}$ & 1 & 0 & $(0,0)$ & 0\tabularnewline
\hline 
$|\mathcal{K}|=t$ & $k\in\mathcal{K}$ & 0 & 0 & $(0,0)$ & 0\tabularnewline
\hline 
$|\mathcal{K}|=t$ & $k\in\mathcal{K}$ & 1 & 0 & $(0,0)$ & 0\tabularnewline
\hline 
$|\mathcal{K}|=t$ & $k\notin\mathcal{K}$ & 0 & 0 & $(0,\texttt{fail})$ & 0\tabularnewline
\hline 
$|\mathcal{K}|=t$ & $k\notin\mathcal{K}$ & 1 & 1 & $(0,\texttt{fail})$ & 1\tabularnewline
\hline 
$|\mathcal{K}|=t+1$ & $k\in\mathcal{K}$ & 0 & 0 & $(0,\texttt{fail})$ & 0\tabularnewline
\hline 
$|\mathcal{K}|=t+1$ & $k\in\mathcal{K}$ & 1 & 1 & $(0,\texttt{fail})$ & 1\tabularnewline
\hline 
$|\mathcal{K}|=t+1$ & $k\notin\mathcal{K}$ & 0 & 0 & $(\texttt{fail},\texttt{fail})$ & 0\tabularnewline
\hline 
$|\mathcal{K}|=t+1$ & $k\notin\mathcal{K}$ & 1 & 1 & $(\texttt{fail},\texttt{fail})$ & 1\tabularnewline
\hline 
$|\mathcal{K}|>t+1$ & $k\in\{1,2,\dots,n\}$ & 0 & 0 & $(\texttt{fail},\texttt{fail})$ & 0\tabularnewline
\hline 
$|\mathcal{K}|>t+1$ & $k\in\{1,2,\dots,n\}$ & 1 & 1 & $(\texttt{fail},\texttt{fail})$ & 1\tabularnewline
\hline 
\end{tabular}
\end{table*}

\subsection{Extrinsic Message Passing \label{subsec:LowComplexEMP}}

In the IMP message-passing rule (\ref{eq: IterHDD2}), the computation of the output message $\mu_{i,j}^{(\ell)}$ passed from $j$ to $i$ uses the input message $\nu_{i,j}^{(\ell)}$. This violates the principle of using only extrinsic information in message-passing rules. The decoding algorithm proposed in Section \ref{sec: Ensembles} can rectify this problem. We note that the messages in the EMP decoder are denoted by $\hat{\nu}_{i,j}^{(\ell)}$ and $\hat{\mu}_{i,j}^{(\ell)}$ to distinguish them from the IMP decoder.

Let $\hat{\nu}_{i,j}^{(\ell)}\in\{0,1\}$ be the message passed by the EMP algorithm from the $i$-th bit node to the $j$-th constraint node and let $\hat{\boldsymbol{\nu}}_{j}^{(\ell)}\triangleq(\hat{\nu}_{\sigma_{j}(1),j}^{(\ell)},\dots,\hat{\nu}_{\sigma_{j}(n),j}^{(\ell)})$ be the collection of all input messages to the $j$-th constraint node in the $\ell$-th iteration. To compute the EMP message $\hat{\mu}_{i,j}^{(\ell)}\triangleq\big(\hat{\mu}_{i,j,0}^{(\ell)},\hat{\mu}_{i,j,1}^{(\ell)}\big)$ from the $j$-th constraint node to the $i$-th bit node, BDD is performed twice. In the $\ell$-th iteration, similar to the arrangement of the candidate decoding vector in (\ref{eq: Candidate}), we first define 
\[
\hat{\bm{\nu}}_{j,k,0}^{(\ell)}\triangleq(\hat{\nu}_{\sigma_{j}(1),j}^{(\ell)},\dots,\hat{\nu}_{\sigma_{j}(k-1),j}^{(\ell)},0,\hat{\nu}_{\sigma_{j}(k+1),j}^{(\ell)},\dots,\hat{\nu}_{\sigma_{j}(n),j}^{(\ell)})
\]
and 
\[
\hat{\bm{\nu}}_{j,k,1}^{(\ell)}\triangleq(\hat{\nu}_{\sigma_{j}(1),j}^{(\ell)},\dots,\hat{\nu}_{\sigma_{j}(k-1),j}^{(\ell)},1,\hat{\nu}_{\sigma_{j}(k+1),j}^{(\ell)},\dots,\hat{\nu}_{\sigma_{j}(n),j}^{(\ell)}),
\]
and then compute $\hat{\boldsymbol{w}}_{j,k,0}^{(\ell)}=\mathsf{D}(\hat{\boldsymbol{\nu}}_{j,k,0}^{(\ell)})$ and $\hat{\boldsymbol{w}}_{j,k,1}^{(\ell)}=\mathsf{D}(\hat{\boldsymbol{\nu}}_{j,k,1}^{(\ell)})$, respectively. Based on $\hat{\boldsymbol{w}}_{j,k,0}^{(\ell)}$ and $\hat{\boldsymbol{w}}_{j,k,1}^{(\ell)}$ computed at the $j$-th constraint node, the messages $\hat{\mu}_{\sigma_{j}(k),j,0}^{(\ell)}$ and $\hat{\mu}_{\sigma_{j}(k),j,1}^{(\ell)}$ is assigned, respectively, by 
\[
\hat{\mu}_{\sigma_{j}(k),j,0}^{(\ell)}=\begin{cases}
\left[\hat{\boldsymbol{w}}_{j,k,0}^{(\ell)}\right]_{k} & \mbox{ if }\mathsf{D}\big(\hat{\bm{\nu}}_{j,k,0}^{(\ell)}\big)\neq\texttt{fail}\\
\texttt{fail} & \mbox{ otherwise,}
\end{cases}
\]
and
\[
\hat{\mu}_{\sigma_{j}(k),j,1}^{(\ell)}=\begin{cases}
\left[\hat{\boldsymbol{w}}_{j,k,1}^{(\ell)}\right]_{k} & \mbox{ if }\mathsf{D}\big(\hat{\bm{\nu}}_{j,k,1}^{(\ell)}\big)\neq\texttt{fail}\\
\texttt{fail} & \mbox{ otherwise.}
\end{cases}
\]
One can see that the message $\hat{\mu}_{i,j}^{(\ell)}$ will be in the set $\{(0,0),(1,1),(0,1),(0,\texttt{fail}),(\texttt{fail},1),(\texttt{fail},\texttt{fail})\}$. We recall that $\nbr(i)=\{j,j'\}$. The message-passing rule for the $i$-th bit node is
\begin{align}
\hat{\nu}_{i,j'}^{(\ell+1)} & \triangleq\begin{cases}
0 & \mbox{if }\hat{\mu}_{i,j}^{(\ell)}=(0,0)\vspace{1ex}\\
1 & \mbox{if }\hat{\mu}_{i,j}^{(\ell)}=(1,1)\\
r_{i} & \mbox{otherwise}.
\end{cases}\label{eq: mphdd2}
\end{align}
The iteration is initialized by setting $\hat{\nu}_{i,j}^{(0)}=r_{i}$ for each bit node $i$ and all $j\in\nbr(i)$. 

We introduce Table \ref{tab:EquivalentTable} to show the equivalence of the iterative HDD algorithm in Section \ref{subsec: iterBDD} and the proposed EMP algorithm. For the $j$-th constraint node, recall that $\hat{\boldsymbol{\nu}}_{j}^{(\ell)}$ is the vector of all input messages to the $j$-th constraint node in the $\ell$-th iteration. Let $\mathcal{K}\triangleq\{k\in\left\{ 1,2,\dots,n\right\} :\hat{\nu}_{\sigma_{j}(k),j}^{(\ell)}=1\}$ be the indices of the sockets that have input message 1. Since BDD satisfies the symmetry property, we assume that the decoder returns an all-zero codeword when the number of error is less than or equal to $t$. Thus, the combinations of the number of input errors, $|\mathcal{K}|$, and the socket of the outgoing message, $k$, listed in the table can cover all possible results of a component code decoder. Note that $\nu_{\sigma_{j}(k),j'}^{(\ell+1)}$ is the iterative HDD message defined in (\ref{eq:MessagePassingRule}), and $\hat{\nu}_{\sigma_{j}(k),j'}^{(\ell+1)}$ is defined in (\ref{eq: mphdd2}). It is clear from Table \ref{tab:EquivalentTable} that the messages of these two algorithms are identical in all cases.

By replacing the $k$-th element of $\hat{\boldsymbol{\nu}}_{j}^{(\ell)}$ with both $0$ and $1$, the computed output $\hat{\mu}_{\sigma_{j}(k),j}^{(\ell)}$ remains independent of the incoming message $\hat{\nu}_{\sigma_{j}(k),j}^{(\ell)}$ on that edge. Therefore, only extrinsic information is used in the computation of the output message on the $(\sigma_{j}(k),j)$ edge from the $j$-th constraint node. The output message from a bit node depends only on the channel observation and the input from the other edge. Therefore, this defines an extrinsic message-passing algorithm with hard-decision messages.

\subsection{Low-Complexity EMP Algorithm}

\begin{algorithm}[b]
Iteration $\ell$: For each constraint node $j$,
\begin{itemize}
\item Compute $\boldsymbol{w}=\mathsf{D}(\hat{\boldsymbol{\nu}}_{j}^{(\ell)})$.
\item For $k=1,\dots,n$,

\begin{itemize}
\item if $d_{H}(\hat{\boldsymbol{\nu}}_{j}^{(\ell)},\boldsymbol{w})>t$, then $\hat{\nu}_{\sigma_{j}(k),j'}^{(\ell+1)}=r_{i_{k}}$
\item elseif $d_{H}(\hat{\boldsymbol{\nu}}_{j}^{(\ell)},\boldsymbol{w})<t$, then $\hat{\nu}_{\sigma_{j}(k),j'}^{(\ell+1)}=w_{k}$
\item else $\hat{\nu}_{\sigma_{j}(k),j'}^{(\ell+1)}=\left(\big(1-\hat{\nu}_{\sigma_{j}(k),j}\big)\big(r_{\sigma_{j}(k)}\,\texttt{OR}\;w_{k}\big)\right)$\\
\hspace*{56mm}$\texttt{OR}\,(r_{\sigma_{j}(k)}w_{k}).$
\end{itemize}
\caption{The low-complexity EMP algorithm \label{alg: lowCplx}}
\end{itemize}
\end{algorithm}

As described above, the EMP algorithm needs to run the BDD algorithm $2n$ times to compute the output messages from a single constraint node. The primary purpose of that description was to demonstrate that the algorithm is indeed an EMP algorithm. Now, we show that exactly the same outputs can be computed with a single decode and some post processing. In the $\ell$-th iteration, let $\boldsymbol{w}\triangleq\mathsf{D}(\hat{\boldsymbol{\nu}}_{j}^{(\ell)})$ be the output of the BDD at the $j$-th constraint node with $\hat{\boldsymbol{\nu}}_{j}^{(\ell)}$ as an input. In this case, we will see that one can calculate $\hat{\boldsymbol{\nu}}_{j}^{(\ell+1)}$ directly from $\hat{\boldsymbol{\nu}}_{j}^{(\ell)}$. In this section, the $\hat{\mu}_{\sigma_{j}(k),j}^{(\ell)}$ messages are used only to explain the correctness of the simplified algorithm. Consider the following facts.
\begin{fact}
If $\boldsymbol{w}=\textup{\texttt{fail}}$, then at least one element of $\hat{\mu}_{\sigma_{j}(k),j}^{(\ell)}$ will be a $\textup{\texttt{fail}}$ for all $k=1,2,\dots,n$. By (\ref{eq: mphdd2}), one can show that $\hat{\nu}_{\sigma_{j}(k),j'}^{(\ell+1)}=r_{\sigma_{j}(k)}$ for all $k=1,2,\dots,n$. \label{fct: 1}
\end{fact}

\begin{fact}
If $\boldsymbol{w}\neq\textup{\texttt{fail}}$ and $d_{H}(\hat{\boldsymbol{\nu}}_{j}^{(\ell)},\boldsymbol{w})<t$, then one can show that $\hat{\mu}_{\sigma_{j}(k),j,0}^{(\ell)}=\hat{\mu}_{\sigma_{j}(k),j,1}^{(\ell)}$ for all $k=1,2,\dots,n$. Thus, we have\textup{ $\hat{\nu}_{\sigma_{j}(k),j'}^{(\ell+1)}=\hat{\mu}_{\sigma_{j}(k),j,0}^{(\ell)}$ }for all $k=1,2,\dots,n$.\label{fct: 2}
\end{fact}

\begin{fact}
If $\boldsymbol{w}\neq\textup{\texttt{fail}}$ and $d_{H}(\hat{\boldsymbol{\nu}}_{j}^{(\ell)},\boldsymbol{w})=t$, then first suppose that $w_{k}=\hat{\nu}_{\sigma_{j}(k),j}^{(\ell)}$ for some $k$. One can see that $\hat{\mu}_{\sigma_{j}(k),j}^{(\ell)}$ must not be $(0,0)$ or $(1,1)$. Thus, we know $\hat{\nu}_{\sigma_{j}(k),j'}^{(\ell+1)}=r_{\sigma_{j}(k)}$. On the other hand, suppose that $w_{k}\neq\hat{\nu}_{\sigma_{j}(k),j}^{(\ell)}$. One can show that $\hat{\mu}_{\sigma_{j}(k),j,0}^{(\ell)}=\hat{\mu}_{\sigma_{j}(k),j,1}^{(\ell)}=w_{k}$. Therefore, $\hat{\nu}_{\sigma_{j}(k),j'}^{(\ell+1)}=w_{k}$.\label{fct: 3}
\end{fact}
Using these facts, we define the low-complexity EMP decoder in Algorithm \ref{alg: lowCplx}. Since the distance $d_{H}(\hat{\boldsymbol{\nu}}_{j}^{(\ell)},\boldsymbol{w})$ is automatically obtained while performing BDD, the additional computation for the EMP algorithm is just a little Boolean logic. Still, the natural data-flow implementation of the conventional decoder is very simple and the EMP algorithm does require additional control logic and memory.

\begin{table}[t]
\begin{centering}
\caption{\vspace{0.5mm} The possible values of $\nu_{i,j'}^{(\ell+1)}$ with input vectors $\mu_{i,j'}^{(\ell-1)}$ when $\bm{c}=\protect\Dec(\bm{\nu}_{j,k,0}^{(\ell)})$ and $\bm{c}'=\protect\Dec(\bm{\nu}_{j,k,1}^{(\ell)})$ are codewords, where $\bm{\nu}_{j,k,0}^{(\ell)}$ and $\bm{\nu}_{j,k,1}^{(\ell)}$ are defined in (\ref{eq: vjk0}) and (\ref{eq: vjk1}), respectively. \label{tab: DependentExample}}
\par\end{centering}
\centering{}%
\begin{tabular}{|c|c|c|c|}
\hline 
$\mu_{i,j'}^{(\ell-1)}$ & $\nu_{i,j}^{(\ell)}$ & $\mu_{i,j}^{(\ell)}$ & $\nu_{i,j'}^{(\ell+1)}$\tabularnewline
\hline 
\hline 
$(0,1)$ & $0$ & $(0,1)$ & $0$\tabularnewline
\hline 
$(1,1)$ & $1$ & $(1,1)$ & $1$\tabularnewline
\hline 
\end{tabular}
\end{table}

\begin{rem}
While preparing this extended version of our earlier work \cite{Jian-isit12}, we discovered that Miladinovic and Fossorier also proposed an iterative HDD algorithm for general product codes \cite{Miladinovic-com08}. We briefly describe their algorithm as follows. For an edge $(i,j)$ connecting the bit node $i$ and the constraint node $j$, let $j'=\nbr(i)\setminus j$, $i=\sigma_{j}(k)$, and $\bm{\nu}_{j}^{(\ell)}\triangleq(\nu_{\sigma_{j}(1),j}^{(\ell)},\nu_{\sigma_{j}(2),j}^{(\ell)}\dots,\nu_{\sigma_{j}(n),j}^{(\ell)})$. Note that the $k$-th element of $\bm{\nu}_{j}^{(\ell)}$ is $\nu_{i,j}^{(\ell)}$. The message passed by the constraint node $j$, denoted by $\mu_{i,j}^{(\ell)}$, consists of two elements $\mu_{i,j}^{(\ell)}\triangleq(\Dec_{k}(\bm{\nu}_{j}^{(\ell)}),s^{(\ell)}),$ where $s^{(\ell)}=1$ if the decoding at the $j$-th constraint node has succeeded; otherwise, $s^{(\ell)}=0$. At the $i$-th bit node, the message $\nu_{i,j'}^{(\ell+1)}$ is updated by 
\begin{align}
\nu_{i,j'}^{(\ell+1)} & =\left(1-s^{(\ell)}\right)r_{i}+s\Dec_{k}\left(\bm{\nu}_{j}^{(\ell)}\right).\label{eq: TCOM08}
\end{align}
One can see that the proposed algorithm is similar to the iterative HDD algorithm proposed. However, the outputs of the two iterative HDD algorithm are different when $t=\frac{d_{\min}-1}{2}$, $\bm{c}=\Dec(\bm{\nu}_{j}^{(\ell)})$, $d_{H}(\bm{c},\bm{\nu}_{j}^{(\ell)})=t$, and $c_{k}=\nu_{i,j}^{(\ell)}\neq r_{i}$. For the proposed iterative HDD and the vector $\bm{v}_{i,j}^{(\ell)}$ defined in (\ref{eq: Candidate}), we know that $\Dec(\bm{v}_{i,j}^{(\ell)})$ will be $r_{i}$, but $(1-s)r_{i}+s\Dec_{k}(\bm{\nu}_{j}^{(\ell)})=c_{k}$. Moreover, we notice that $\nu_{i,j'}^{(\ell+1)}$ in the update equation (\ref{eq: TCOM08}) will depend on $\mu_{i,j'}^{(\ell-1)}$. In the $\ell$-th iteration, we define two vectors 
\begin{align}
\!\!\!\!\bm{\nu}_{j,k,0}^{(\ell)} & \!\!\triangleq\!\left(\!\nu_{\sigma_{j}(1),j}^{(\ell)},\!\cdots\!,\nu_{\sigma_{j}(k\!-\!1),j}^{(\ell)},0,\nu_{\sigma_{j}(k\!+\!1),j}^{(\ell)},\!\cdots\!,\nu_{\sigma_{j}(n),j}^{(\ell)}\!\right)\!\!\label{eq: vjk0}
\end{align}
and 
\begin{align}
\!\!\!\!\bm{\nu}_{j,k,1}^{(\ell)} & \!\!\triangleq\!\left(\!\nu_{\sigma_{j}(1),j}^{(\ell)},\!\cdots\!,\nu_{\sigma_{j}(k\!-\!1),j}^{(\ell)},1,\nu_{\sigma_{j}(k\!+\!1),j}^{(\ell)},\!\cdots\!,\nu_{\sigma_{j}(n),j}^{(\ell)}\!\right)\!.\!\!\!\label{eq: vjk1}
\end{align}
Suppose that the decoder outputs $\bm{c}=\Dec(\bm{\nu}_{j,k,0}^{(\ell)})$ and $\bm{c}'=\Dec(\bm{\nu}_{j,k,1}^{(\ell)})$ are both codewords. It is clear that $d_{H}(\bm{c},\bm{\nu}_{j,k,0}^{(\ell)})=d_{H}(\bm{c}',\bm{\nu}_{j,k,1}^{(\ell)})=t$, $\bm{c}_{k}=0$, and $\bm{c}'_{k}=1$. Also, we know $s=1$ for the decoding of both vectors. The possible values of $\nu_{i,j'}^{(\ell+1)}$ are listed in Table \ref{tab: DependentExample}. Since the values of $\nu_{i,j'}^{(\ell+1)}$ for different $\mu_{i,j'}^{(\ell-1)}$ are different, we observe that $\nu_{i,j'}^{(\ell+1)}$ is not independent of $\mu_{i,j'}^{(\ell-1)}$. Thus, it is not clear if their DE analysis can be rigorously justified.
\end{rem}

\section{Numerical Results and Comparison}

In the following numerical results, the iterative HDD threshold of $(\set C,m,L,w)$ spatially-coupled GLDPC ensemble with $L=1025$, and $w=16$ are considered. In Table \ref{tab: Table}, the thresholds of the ensembles are shown in terms of the average number of error bits attached to a code-constraint node. Let $p_{n,t}^{*}$ be the iterative HDD threshold of the spatially-coupled GLDPC ensemble based on a $(n,k,2t+1)$ binary primitive BCH component code, and $\tilde{p}_{n,t}^{*}$ be the iterative HDD threshold of the spatially-coupled GLDPC ensemble based on the $(n,k-1,2t+2)$ even-weight subcode. Then, we define $a_{n,t}^{*}\triangleq np_{n,t}^{*}$ and $\tilde{a}_{n,t}^{*}\triangleq n\tilde{p}_{n,t}^{*}$ to be the thresholds in terms of the average number of error bits attached to a component code. In the high-rate scaling limit, we let $\rho_{t}^{*}$ and $\tilde{\rho}_{t}^{*}$ denote the iterative HDD thresholds of the ensembles based on primitive BCH component codes and their even-weight subcodes, respectively. Moreover, the threshold of HDD without miscorrection, $\hat{\rho}_{t}^{*}$, is shown in Table \ref{tab: Table} along with the coupled threshold, $\hat{\rho}_{t}^{**}$, of iterative HDD without miscorrection from (\ref{eq: HighRateNoMiss}).

From Table \ref{tab: Table}, one can observe that the thresholds ($\rho_{t}^{*}$, $\tilde{\rho}_{t}^{*}$ and $\hat{\rho}_{t}^{*}$ ) of the spatially-coupled ensemble with primitive BCH component codes or the even-weight subcodes approach to $2t$ as $t$ increases. This verifies the results predicted by Lemma~\ref{cor: high-rate_thresh_no_miscorr} and the vanishing impact of miscorrection predicted by Lemma~\ref{lem: high-rate_thresh_w_miscorr}. 

\begin{table}[t]
\caption{\vspace{0.5mm} \label{tab: Table} Iterative HDD thresholds of $(\protect\set C,m,1025,16)$ spatially-coupled GLDPC ensemble with binary primitive BCH codes}

\centering{}%
\begin{tabular}{c|ccccc}
$t$ & 3 & 4 & 5 & 6 & 7\tabularnewline
\hline 
$a_{255,t}^{*}$ & 5.432 & 7.701 & 9.818 & 11.86 & 13.87\tabularnewline
$a_{511,t}^{*}$ & 5.417 & 7.665 & 9.811 & 11.86 & 13.85\tabularnewline
$a_{1023,t}^{*}$ & 5.401 & 7.693 & 9.821 & 11.87 & 13.88\tabularnewline
$\rho_{t}^{*}$ & 5.390 & 7.688 & 9.822 & 11.91 & 13.93\tabularnewline
\hline 
$\tilde{a}_{255,t}^{*}$ & 5.610 & 7.752 & 9.843 & 11.88 & 13.87\tabularnewline
$\tilde{a}_{511,t}^{*}$ & 5.570 & 7.767 & 9.811 & 11.86 & 13.85\tabularnewline
$\tilde{a}_{1023,t}^{*}$ & 5.606 & 7.765 & 9.841 & 11.88 & 13.88\tabularnewline
$\tilde{\rho}_{t}^{*}$ & 5.605 & 7.761 & 9.840 & 11.91 & 13.93\tabularnewline
\hline 
$\hat{\rho}_{t}^{*}$ & 5.735 & 7.813 & 9.855 & 11.91 & 13.93\tabularnewline
$\hat{\rho}_{t}^{**}$ & 5.754 & 7.843 & 9.896 & 11.93 & 13.95\tabularnewline
\hline 
\end{tabular}
\end{table}

\vspace{0mm}

\section{Conclusion}

\vspace{0mm}

The iterative HDD of GLDPC ensembles, based on on $t$-error correcting block codes, is analyzed with and without spatial coupling. Using DE analysis, noise thresholds are computed for a variety of component codes and decoding assumptions. In particular, the case of binary primitive BCH component-codes is considered along with their even-weight subcodes. For these codes, the miscorrection probability is characterized and included in the DE analysis. Scaled DE recursions are also computed for the high-rate limit. When miscorrection is neglected, the resulting recursion for the basic ensemble matches the results of \cite{Justesen-commag10,Justesen-toc11}. It is also proven that iterative HDD threshold of the spatially-coupled GLDPC ensemble can approach capacity in high-rate regime. Finally, numerical results are presented that both verify the theoretical results and demonstrate the effectiveness of these codes for high-speed communication systems.

\vspace{0.0mm}

\appendices{}

\section{Proof of Lemma \ref{lem:LimitOfPn(i)andnQn(i)}\label{app: pfLemLimOfPnnQn}}
\begin{IEEEproof}
For a fixed $i\geq t$, one can use (\ref{eq: WeightDistrib}) to rewrite the last term of (\ref{eq: P(i)}) as
\[
\left(1+O\left(n^{-0.1}\right)\right)\sum_{\delta=1}^{t}\sum_{j=0}^{\delta-1}F(n,i,\delta,j),
\]
where 
\begin{align*}
 & F(n,i,\delta,j)\triangleq\frac{n-l(i,\delta,j)}{n}2^{-mt}\binom{n}{l(i,\delta,j)}\\
 & \qquad\times\binom{l(i,\delta,j)}{l(i,\delta,j)-j}\binom{n-l(i,\delta,j)-1}{\delta-1-j}\binom{n-1}{i}^{-1}.
\end{align*}
Thus, $P_{n}(i)$ is given by 
\begin{align}
P_{n}(i) & =1-\left(1+O\left(n^{-0.1}\right)\right)\sum_{\delta=1}^{t}\sum_{j=0}^{\delta-1}F(n,i,\delta,j).\label{eq: Pn(i)byF}
\end{align}
By the fact that $2^{m}=n+1,$ it is easy to verify that
\begin{align}
 & F(n,i,\delta,j)\nonumber \\
 & \quad=\frac{n-l(i,\delta,j)}{n}\frac{1}{(n+1)^{t}}\frac{n!}{(n-l(i,\delta,j))!l(i,\delta,j)!}\nonumber \\
 & \qquad\times\frac{l(i,\delta,j)!}{(l(i,\delta,j)-j)!j!}\frac{(n-l(i,\delta,j)-1)!}{(n-l(i,\delta,j)-\delta+j)!(\delta-1-j)!}\nonumber \\
 & \qquad\qquad\times\frac{i!(n-1-i)!}{(n-1)!}\nonumber \\
 & \quad=\frac{1}{(n+1)^{t}}\frac{i!}{(l(i,\delta,j)-j)!}\nonumber \\
 & \qquad\qquad\times\frac{(n-1-i)!}{(n-l(i,\delta,j)-\delta+j)!}\frac{1}{j!(\delta-1-j)!}\label{eq: tmp1}\\
 & \quad=\frac{1}{(n+1)^{t}}\frac{i!}{(i-\delta+j+1)!}\nonumber \\
 & \qquad\qquad\times\frac{(n-1-i)!}{(n-i-j-1)!}\frac{1}{j!(\delta-1-j)!},\label{eq: ExpandF}
\end{align}
where (\ref{eq: ExpandF}) is obtained by substituting (\ref{eq: l(i,d,j)}) into (\ref{eq: tmp1}). When $j<\delta-1$, (\ref{eq: ExpandF}) can be written as\\
\begin{equation}
\!\!\!\!\!\!\frac{1}{j!(\delta\!-\!1\!-\!j)!(n\!+\!\!1)^{t-\delta+1}}\!\!\left(\prod_{k=0}^{\delta-j-2}\!\!\frac{i\!-\!k}{n\!+\!1}\!\!\right)\!\!\!\left(\prod_{k'=0}^{j-1}\!\!\frac{n\!-\!1\!-\!i\!-\!k'}{n\!+\!1}\!\right)\!\!.\!\!\!\label{eq: ExpandF-1}
\end{equation}
On the other hand, when $j=\delta-1$, (\ref{eq: ExpandF}) becomes 
\begin{equation}
\frac{1}{(n+1)^{t-\delta+1}}\frac{1}{(\delta-1)!}\left(\prod_{k'=0}^{\delta-2}\frac{n-1-i-k'}{n+1}\right).\label{eq: ExpandF-2}
\end{equation}
Substituting (\ref{eq: ExpandF-1}) and (\ref{eq: ExpandF-2}) into (\ref{eq: Pn(i)byF}), we have 
\begin{align*}
 & P_{n}(i)=1-\left(1+O\left(n^{-0.1}\right)\right)\\
 & \qquad\times\left(\sum_{j=0}^{t-1}F(n,i,t,j)+\sum_{\delta=1}^{t-1}\sum_{j=0}^{\delta-1}F(n,i,\delta,j)\right)\\
 & \;\;=1-\left(1+O\left(n^{-0.1}\right)\right)(n+1)^{-1}\\
 & \qquad\times\left(\frac{1}{(t-1)!}\prod_{k'=0}^{t-2}\frac{n-1-i-k'}{n+1}+\sum_{j=0}^{t-2}\frac{1}{j!(t-1-j)!}\right.\\
 & \qquad\left.\times\left(\prod_{k=0}^{t-j-2}\frac{i-k}{n+1}\right)\!\left(\prod_{k'=0}^{j-1}\frac{n-1-i-k'}{n+1}\right)\!\right)\!+\!O\left(n^{-2}\right)\\
 & \ \;>1-\left(1+O\left(n^{-0.1}\right)\right)n^{-1}\\
 & \qquad\quad\times\left(\sum_{j=0}^{t-1}\frac{1}{j!(t-j-1)!}\right)+O\left(n^{-2}\right).
\end{align*}
Since $\sum_{j=0}^{t-1}\frac{1}{j!(t-j-1)!}\leq2$, it follows that $\lim_{n\rightarrow\infty}P_{n}(i)=1$.

For the analysis of $nQ_{n}(i)$, we also define 
\begin{align}
K & (n,i,\delta,j)\triangleq\frac{l(i,\delta,j)-1}{n}2^{-mt}\binom{n}{l(i,\delta,j)-1}\nonumber \\
 & \quad\;\times\binom{l(i,\delta,j)-2}{l(i,\delta,j)\!-\!j\!-\!1}\!\!\binom{n\!-\!l(i,\delta,j)\!+\!1}{\delta-1}\!\!\binom{n\!-\!1}{i}^{\!-1}\label{eq: K}
\end{align}
so that
\begin{align}
nQ_{n}(i) & =\sum_{\delta=1}^{t}\sum_{j=0}^{\delta}\left(1+O\left(n^{-0.1}\right)\right)nK(n,i,\delta,j).\label{eq: nQn(i)}
\end{align}
Now, we will show that $nQ_{n}(i)$ is bounded by a constant independent of $n$ for all $t+1\leq i\leq n-t-1$. Using a simplification similar to the one used above for (\ref{eq: ExpandF}), one finds that %
\begin{align}
 & nK(n,i,\delta,j)\nonumber \\
 & \ =\frac{n}{(n+1)^{t}}\frac{i!}{(i-\delta+j)!}\frac{(n-1-i)!}{(n-i-j)!}\frac{1}{(j-1)!(\delta-j)!}.\!\!\!\label{eq: ExpandednK}
\end{align}
When $j<\delta$, the RHS of (\ref{eq: ExpandednK}) can be simplified by
\begin{align}
 & \frac{1}{(j-1)!(\delta-j)!}\frac{n}{(n+1)^{t-\delta+1}}\nonumber \\
 & \quad\times\left(\prod_{k=0}^{\delta-j-1}\frac{i-k}{n+1}\right)\left(\prod_{k'=0}^{j-2}\frac{n-1-i-k'}{n+1}\right).\label{eq: ExpandednK-1}
\end{align}
On the other hand, when $j=\delta$, (\ref{eq: ExpandednK}) becomes 
\begin{equation}
\frac{1}{(\delta-1)!}\frac{n}{(n+1)^{t-\delta+1}}\left(\prod_{k'=0}^{\delta-2}\frac{n-1-i-k'}{n+1}\right).\label{eq: ExpandednK-2}
\end{equation}
In both cases, it is easy to verify that (\ref{eq: ExpandednK-1}) and (\ref{eq: ExpandednK-2}) are upper bounded by $1$. Therefore, $nQ_{n}(i)$ is bounded by a constant for all $t+1\leq i\leq n-t-1$.%

To show (\ref{eq: finite_n}), we first introduce some upper bounds on $nK(n,i,\delta,j)$. When $j$ and $\delta$ satisfy $\delta<t$ and $0\leq j\leq\delta$, both (\ref{eq: ExpandednK-1}) and (\ref{eq: ExpandednK-2}) imply that $nK(n,i,\delta,j)<\frac{n}{n^{t-\delta+1}}\leq\frac{1}{n}$ for all $t+1\leq i\leq n-t-1$.%
{} Also, by substituting $\delta=t$ into (\ref{eq: ExpandednK-2}), we have $nK(n,i,t,t)<\frac{1}{(t-1)!}$ for all $t+1\leq i\leq n-t-1$. When $\delta=t$ and $j<\delta$, we consider (\ref{eq: ExpandednK-1}) on the interval $t+1\leq i\leq\lfloor\sqrt{n}\rfloor$. Analyzing (\ref{eq: ExpandednK-1}), we see that
\begin{align*}
nK(n,i,t,j) & \leq\frac{\lfloor\sqrt{n}\rfloor}{n+1}<\frac{1}{\sqrt{n}}.
\end{align*}
From (\ref{eq: nQn(i)}), $nQ_{n}(i)$ for $t+1\leq i\leq\lfloor\sqrt{n}\rfloor$ can be upper bounded by
\begin{align}
 & nQ_{n}(i)=\sum_{\delta=1}^{t}\sum_{j=0}^{\delta}\left(1+O\left(n^{-0.1}\right)\right)nK(n,i,\delta,j)\nonumber \\
 & \quad\leq\left(1+O\left(n^{-0.1}\right)\right)\left(\sum_{\delta=1}^{t-1}\sum_{j=0}^{\delta}\frac{1}{n}+\sum_{j=0}^{t-1}\frac{1}{\sqrt{n}}+\frac{1}{(t-1)!}\right)\nonumber \\
 & \quad=\frac{1}{(t-1)!}\left(1+O\left(n^{-0.1}\right)\right).\label{eq: nQnFiniteUB}
\end{align}
Thus, for fixed $i\geq t+1$, we conclude that $\lim_{n\rightarrow\infty}nQ_{n}(i)=\frac{1}{(t-1)!}$. 
\end{IEEEproof}

\section{Proof of Lemma \ref{lem: limEPn_limEQn}\label{app: pf_limEPn_limEQn}}
\begin{IEEEproof}
Since $P_{n}(i)=0$ for $0\leq i\leq t-1$ and $P_{n}(i)<1$ for all $i\geq t$, we know that 
\[
1-E\left[P_{n}(X_{n})\right]\geq\sum_{i=0}^{t-1}\binom{n\!-\!1}{i}\!\left(\frac{\lambda_{n}}{n\!-\!1}\right)^{i}\!\left(\!1\!-\!\frac{\lambda_{n}}{n\!-\!1}\!\right)^{n-1-i},
\]
and Poisson convergence implies $\lim_{n\rightarrow\infty}E\left[P_{n}(X_{n})\right]\leq\phi(\lambda;t-1).$ With the convention that $\binom{n}{k}=0$ if $k>n$, we can fix $T$ and write
\begin{align*}
E\left[P_{n}(X_{n})\right] & \geq\!\sum_{i=t}^{T}\!\binom{n\!-\!1}{i}\!\left(\frac{\lambda_{n}}{n\!-\!1}\right)^{i}\!\left(\!1\!-\!\frac{\lambda_{n}}{n\!-\!1}\!\right)^{n-1-i}\!\!P_{n}(i).
\end{align*}
Again, Poisson convergence implies
\begin{align*}
 & \lim_{n\rightarrow\infty}E\left[P_{n}(X_{n})\right]\\
 & \quad\geq\lim_{n\rightarrow\infty}\!\sum_{i=t}^{T}\!\binom{n\!-\!1}{i}\!\left(\frac{\lambda_{n}}{n\!-\!1}\right)^{i}\!\left(\!1\!-\!\frac{\lambda_{n}}{n\!-\!1}\!\right)^{n-1-i}\!\!P_{n}(i)\\
 & \quad=\phi(\lambda;t-1)-\phi(\lambda;T).
\end{align*}
Since $T$ is arbitrary and by Markov's inequality $\phi(\lambda;T)\leq\frac{\lambda}{T+1}$, it follows that $\lim_{n\rightarrow\infty}E\left[P_{n}(X_{n})\right]=\phi(\lambda;t-1)$.

Since $Q_{n}(i)=0$ for $i\leq t$ and $Q_{n}(i)\leq1$ for all $i\geq t+1$, we can write
\begin{align*}
 & E\left[Q_{n}(X_{n})\right]\\
 & \quad\leq\sum_{i=t+1}^{T}\binom{n-1}{i}\left(\frac{\lambda_{n}}{n-1}\right)^{i}\left(1-\frac{\lambda_{n}}{n-1}\right)^{n-1-i}Q_{n}(i)\\
 & \quad\qquad+\sum_{i=T+1}^{\infty}\binom{n-1}{i}\left(\frac{\lambda_{n}}{n-1}\right)^{i}\left(1-\frac{\lambda_{n}}{n-1}\right)^{n-1-i}.
\end{align*}
Using (\ref{eq: nQnFiniteUB}), we see that $\lim_{n\rightarrow\infty}Q_{n}(i)=0$ for any fixed $i$. Thus, the previous equation shows that $\lim_{n\rightarrow\infty}E\left[Q_{n}(X_{n})\right]\leq\phi(\lambda;T)$. Since $Q_{n}(i)$ is non-negative and $T$ can be chosen arbitrarily large, this implies that $\lim_{n\rightarrow\infty}E\left[Q_{n}(X_{n})\right]=0$.
\end{IEEEproof}

\section{Proof of Lemma \ref{lem: limEnQn}\label{app: pf_limEnQn}}
\begin{IEEEproof}
By using the convention that $\binom{n}{k}=0$ when ever $k>n$, for any $T\geq t$, we can rewrite $E\left[nQ_{n}(X_{n})\right]$ as
\begin{align*}
 & E\left[nQ_{n}(X_{n})\right]\\
 & \quad=\!\sum_{i=t+1}^{T}\!\binom{n\!-\!1}{i}\!\left(\frac{\lambda_{n}}{n\!-\!1}\right)^{i}\!\left(\!1\!-\!\frac{\lambda_{n}}{n\!-\!1}\!\right)^{n-1-i}\!\!nQ_{n}(i)\\
 & \qquad+\!\sum_{i=T+1}^{\infty}\!\binom{n\!-\!1}{i}\!\left(\frac{\lambda_{n}}{n\!-\!1}\right)^{i}\!\left(\!1\!-\!\frac{\lambda_{n}}{n\!-\!1}\!\right)^{n-1-i}\!\!nQ_{n}(i).
\end{align*}
Then, by the Poisson theorem \cite[pp. 113]{Papoulis-2002} and Lemma \ref{lem:LimitOfPn(i)andnQn(i)} we know
\begin{align*}
 & \lim_{n\rightarrow\infty}E\left[nQ_{n}(X_{n})\right]\\
 & \quad\geq\lim_{n\rightarrow\infty}\sum_{i=t+1}^{T}\!\!\binom{n\!-\!1}{i}\left(\frac{\lambda_{n}}{n\!-\!1}\right)^{\!i}\!\!\left(\!1\!-\!\frac{\lambda_{n}}{n\!-\!1}\!\right)^{\!n-1-i}\!nQ_{n}(i)\\
 & \quad=\frac{1}{(t-1)!}\left(\phi(\lambda;t)-\phi(\lambda;T)\right).
\end{align*}
Since $nQ_{n}(i)$ is bounded for $t+1\leq i\leq n-2-t$ according to Lemma \ref{lem:LimitOfPn(i)andnQn(i)}, there exists a constant $0<C<\infty$ independent of $n$ such that $nQ_{n}(i)\leq C$ for $t+1\leq i\leq n-2-t$. Also, we know $nQ_{n}(i)=0$ for $0\leq i\leq t$ and $nQ_{n}(i)=n$ for $n-t-1\leq i\leq n-1$. Thus, $E\left[nQ_{n}(X_{n})\right]$ can be upper bounded by 
\begin{align}
 & E\left[nQ_{n}(X_{n})\right]\nonumber \\
 & \quad\leq\sum_{i=t+1}^{T}\!\binom{n-1}{i}\!\left(\frac{\lambda_{n}}{n-1}\right)^{i}\!\left(1-\frac{\lambda_{n}}{n-1}\right)^{n-1-i}nQ_{n}(i)\nonumber \\
 & \qquad+C\sum_{i=T+1}^{\infty}\!\binom{n-1}{i}\!\left(\frac{\lambda_{n}}{n-1}\right)^{i}\!\left(1-\frac{\lambda_{n}}{n-1}\right)^{n-1-i}\nonumber \\
 & \qquad\quad+\!n\!\!\sum_{i=n-t-1}^{n-1}\!\!\binom{n\!-\!1}{i}\!\left(\frac{\lambda_{n}}{n\!-\!1}\right)^{\!i}\!\left(\!1\!-\!\frac{\lambda_{n}}{n\!-\!1}\!\right)^{\!n-1-i}\!.\!\label{eq: EnQnUB-1}
\end{align}
By the Chernoff bound of the binomial distribution ${\rm Bi}(n-1,\frac{\lambda_{n}}{n-1})$, the last term of (\ref{eq: EnQnUB-1}) can be upper bounded by

\begin{align}
 & n\!\!\!\sum_{i=n-t-1}^{n-1}\!\!\binom{n-1}{i}\!\left(\frac{\lambda_{n}}{n-1}\right)^{i}\!\!\left(1-\frac{\lambda_{n}}{n-1}\right)^{n-1-i}\nonumber \\
 & \qquad\qquad\leq\!n\left(\frac{e\lambda_{n}}{n-1-t}\right)^{n-t-1}\!\!e^{-\lambda_{n}}.\label{eq: BinomTailBound1}
\end{align}
Since $\lambda_{n}\rightarrow\lambda$ and $\lambda<\infty$, we know that there exists a $N_{0}>0$ and a $\epsilon>0$ such that $\lambda_{n}\leq\lambda+\epsilon$ whenever $n>N_{0}$. Thus, we have 
\begin{align*}
0 & \leq\lim_{n\rightarrow\infty}n\left(\frac{e\lambda_{n}}{n-1-t}\right)^{n-t-1}e^{-\lambda_{n}}\\
 & \leq e^{-\lambda}\lim_{n\rightarrow\infty}n\left(\frac{e(\lambda+\epsilon)}{n-1-t}\right)^{n-t-1}=0,
\end{align*}
and 
\begin{align*}
 & \lim_{n\rightarrow\infty}E\left[nQ_{n}(X_{n})\right]\\
 & \quad\leq\lim_{n\rightarrow\infty}\!\sum_{i=t+1}^{T}\!\!\binom{n\!-\!1}{i}\!\!\left(\frac{\lambda_{n}}{n\!-\!1}\right)^{\!i}\!\!\left(\!1\!-\!\frac{\lambda_{n}}{n\!-\!1}\!\right)^{\!n-1-i}\!\!nQ_{n}(i)\\
 & \qquad+C\!\lim_{n\rightarrow\infty}\!\sum_{i=T+1}^{\infty}\!\!\binom{n\!-\!1}{i}\!\!\left(\frac{\lambda_{n}}{n\!-\!1}\right)^{\!i}\!\!\left(\!1\!-\!\frac{\lambda_{n}}{n\!-\!1}\!\right)^{\!n-1-i}\\
 & \quad=\frac{1}{(t-1)!}\phi(\lambda;t)+\left(C-\frac{1}{(t-1)!}\right)\phi(\lambda;T).
\end{align*}
Therefore, we have 
\begin{align*}
 & \frac{1}{(t-1)!}\phi(\lambda;t)-\frac{1}{(t-1)!}\phi(\lambda;T)\\
 & \quad\leq\lim_{n\rightarrow\infty}E\left[nQ_{n}(X_{n})\right]\\
 & \quad\leq\frac{1}{(t-1)!}\phi(\lambda;t)+\left(C-\frac{1}{(t-1)!}\right)\phi(\lambda;T).
\end{align*}
Since $C$ is independent of $n$ and $T$ is arbitrary, we know $\lim_{n\rightarrow\infty}E\left[nQ_{n}(X_{n})\right]=\frac{1}{(t-1)!}\phi(\lambda;t)$. 
\end{IEEEproof}

\section{Proof of Lemma \ref{lem: subcodes}\label{app: pf_subcodes}}
\begin{IEEEproof}
To show the lemma, it suffices to show that $\lim_{n\rightarrow\infty}\tilde{P}_{n}(i)=1$ for $i\geq t$, $\lim_{n\rightarrow\infty}E[\tilde{Q}_{n}(X_{n})]=0$, and 
\begin{align}
\lim_{n\rightarrow\infty}E\left[n\tilde{Q}_{n}(X_{n})\right] & =\begin{cases}
\frac{1}{(t-1)!}\psi\left(\lambda;t\right) & \mbox{ if \ \ensuremath{t\ }is even},\\
\frac{1}{(t-1)!}\varphi\left(\lambda;t\right) & \mbox{ if \ensuremath{\ t\ }is odd}.
\end{cases}\label{eq: tildeQ}
\end{align}
Let $\set L_{1}(i,\delta)\triangleq\{j\in[0,\delta]\mid l(i,\delta,j)=0\ \mbox{mod}\ 2\}$ be the set of all $j$ such that $l(i,\delta,j)$ is an even number. Since $A_{\ell}=0$ for all odd $\ell$, we have
\begin{align*}
\tilde{P}_{n}(i) & =1-\sum_{\delta=1}^{t}\sum_{j\in\set L_{1}(i,\delta)}\frac{n-l(i,\delta,j)}{n}A_{l(i,\delta,j)}\itTheta(n,i,\delta,j)\\
 & \geq P_{n}(i).
\end{align*}
Likewise, we define $\set L_{2}(i,\delta)\triangleq\{j\in[0,\delta]\mid(l(i,\delta,j)-1)=0\ \mbox{mod}\ 2\}$ be the set of all $j$ such that $l(i,\delta,j)-1$ is an even number. Then, 
\begin{align*}
\tilde{Q}_{n}(i) & =\sum_{\delta=1}^{t}\sum_{j\in\set L_{2}(i,\delta)}\frac{l(i,\delta,j)-1}{n}A_{l(i,\delta,j)-1}\itLambda(n,i,\delta,j)\\
 & \leq Q_{n}(i).
\end{align*}
From Lemma \ref{lem:LimitOfPn(i)andnQn(i)} and Lemma \ref{lem: limEPn_limEQn}, we immediately have $\lim_{n\rightarrow\infty}\tilde{P}_{n}(i)=1$ for $i\geq t$, and $\lim_{n\rightarrow\infty}E[\tilde{Q}_{n}(X_{n})]=0$. 

From (\ref{eq: K}) and (\ref{eq: nQn(i)}, we have 
\begin{align*}
n\tilde{Q}_{n}(i) & =\sum_{\delta=1}^{t}\sum_{j\in\set L_{2}(i,\delta)}\left(1+O\left(n^{-0.1}\right)\right)nK(n,i,\delta,j),
\end{align*}
When $t$ is even, and $i$ is odd, one can show that $\set L_{2}(i,t)=\emptyset$. From (\ref{eq: ExpandednK-2}), we know 
\begin{align*}
n\tilde{Q}_{n}(i) & \leq\sum_{\delta=1}^{t-1}\sum_{j\in0}^{\delta}\left(1+O\left(n^{-0.1}\right)\right)nK(n,i,\delta,j)\\
 & =O\left(n^{-1}\right).
\end{align*}
When both $t$ and $i$ are even numbers, one can show $t\in\set L_{2}(i,t)$. By the same argument in the proof of Lemma \ref{lem:LimitOfPn(i)andnQn(i)}, we know $n\tilde{Q}_{n}(i)$ is upper bounded by a constant for all even $i$ with $i\geq t+1$, and for all even $t+1\leq i\leq\sqrt{n}$, $n\tilde{Q}_{n}(i)=\frac{1}{(t-1)!}(1+O(n^{-0.1}))$. Let $\mathbb{N}_{e}$ be the set of even natural numbers, and $\mathbb{N}_{o}$ be the set of odd nature numbers. Then, when $t$ is even, we have 
\begin{align*}
 & E\left[n\tilde{Q}_{n}(X_{n})\right]\\
 & \quad=\!\!\sum_{i\in\mathbb{N}_{e},i\geq t+2}\!\!\binom{n\!-\!1}{i}\!\!\left(\frac{\lambda_{n}}{n\!-\!1}\right)^{\!i}\!\!\left(\!1\!-\!\frac{\lambda_{n}}{n\!-\!1}\!\right)^{\!n-1-i}\!\!n\tilde{Q}_{n}(i)\\
 & \qquad+\!\!\sum_{i\in\mathbb{N}_{o},i\geq t+1}\!\!\binom{n\!-\!1}{i}\!\!\left(\frac{\lambda_{n}}{n\!-\!1}\right)^{\!i}\!\!\left(\!1\!-\!\frac{\lambda_{n}}{n\!-\!1}\!\right)^{\!n-1-i}\!\!n\tilde{Q}_{n}(i)\\
 & \quad=\!\!\sum_{i\in\mathbb{N}_{e},i\geq t+2}\!\!\binom{n\!-\!1}{i}\!\!\left(\frac{\lambda_{n}}{n\!-\!1}\right)^{\!i}\!\!\left(\!1\!-\!\frac{\lambda_{n}}{n\!-\!1}\!\right)^{\!n-1-i}\!\!n\tilde{Q}_{n}(i)\\
 & \qquad+O\left(\frac{1}{n}\right).
\end{align*}
By the same calculation in the proof of Lemma \ref{lem: limEnQn}, we have 
\begin{align*}
 & \lim_{n\rightarrow\infty}E\left[n\tilde{Q}_{n}(X_{n})\right]\\
 & \quad=\frac{1}{(t-1)!}\sum_{i\in\mathbb{N}_{e},i\geq t+2}\frac{\lambda^{i}e^{-\lambda}}{i!}\\
 & \quad=\frac{1}{(t-1)!}\left(\sum_{i\in\mathbb{N}_{e},i\geq t+2}\frac{\lambda^{i}e^{-\lambda}}{i!}-\sum_{i'=0}^{\frac{t}{2}}\frac{\lambda^{2i}e^{-\lambda}}{(2i)!}\right)\\
 & \quad\overset{\mbox{(a)}}{=}\frac{1}{(t-1)!}\left(\frac{\left(1-e^{-2\lambda}\right)}{2}-\sum_{i=0}^{\frac{t}{2}}\frac{\lambda^{2i}e^{-\lambda}}{(2i)!}\right),
\end{align*}
where (a) follows from the fact that $\sum_{i\in\mathbb{N}_{e},i\geq t+2}\frac{\lambda^{i}e^{-\lambda}}{i!}=\frac{1}{2}(1-e^{-2\lambda})$. Thus, we have verified (\ref{eq: tildeQ}) for even $t$.

When $t$ is a odd number, we know $\set L_{2}(i,t)=\emptyset$ for all even $i$. Then, one can have $n\tilde{Q}_{n}(i)<\frac{1}{n}$ for all even $i\geq t+1$, and $n\tilde{Q}_{n}(i)=\frac{1}{(t-1)!}(1+O(n^{-0.1}))$ for all odd $t+1\leq i\leq\sqrt{n}$. Applying the same argument to the case where $t$ is even, we get
\begin{align*}
 & E\left[n\tilde{Q}_{n}(X_{n})\right]\\
 & \quad=\!\!\sum_{i\in\mathbb{N}_{o},i\geq t+2}\!\!\binom{n\!-\!1}{i}\!\!\left(\frac{\lambda_{n}}{n\!-\!1}\right)^{\!i}\!\!\left(\!1\!-\!\frac{\lambda_{n}}{n\!-\!1}\right)^{\!n-1-i}\!\!n\tilde{Q}_{n}(i)\\
 & \qquad+\!\!\sum_{i\in\mathbb{N}_{e},i\geq t+1}\!\!\binom{n\!-\!1}{i}\!\left(\frac{\lambda_{n}}{n\!-\!1}\right)^{\!i}\!\left(\!1\!-\!\frac{\lambda_{n}}{n\!-\!1}\right)^{\!n-1-i}\!\!n\tilde{Q}_{n}(i)\\
 & \quad=\!\!\sum_{i\in\mathbb{N}_{o},i\geq t+2}\!\!\binom{n\!-\!1}{i}\!\!\left(\frac{\lambda_{n}}{n\!-\!1}\right)^{\!i}\!\!\left(\!1\!-\!\frac{\lambda_{n}}{n\!-\!1}\!\right)^{\!n-1-i}\!\!n\tilde{Q}_{n}(i)\\
 & \qquad+O\left(\frac{1}{n}\right).
\end{align*}
Then, 
\begin{align*}
 & \lim_{n\rightarrow\infty}E\left[n\tilde{Q}_{n}(X_{n})\right]\\
 & \quad=\frac{1}{(t-1)!}\sum_{i\in\mathbb{N}_{o},i\geq t+2}\frac{\lambda^{i}e^{-\lambda}}{i!}\\
 & \quad=\frac{1}{(t-1)!}\left(\frac{\left(1+e^{-2\lambda}\right)}{2}-\sum_{i=0}^{\frac{t-1}{2}}\frac{\lambda^{2i+1}e^{-\lambda}}{(2i+1)!}\right).
\end{align*}
This verifies (\ref{eq: tildeQ}) for odd $t$ and completes the proof the lemma.
\end{IEEEproof}

\section{Proof of Lemma \ref{lem: Vn_unique_root} \label{app: pf_Vnx_unique_root}}
\begin{IEEEproof}
From~(\ref{eq: fixed_point_pot_f}) and (\ref{eq: fn_eq_Ix}), we can see that 
\begin{align*}
\hV_{n}'(x) & =-\frac{1}{2}\hfn(z)+\frac{1}{2}x\hfn'(x)\\
 & =-\frac{1}{2B(t,n-t)}\Big(\int_{0}^{x}z^{t-1}(1-z)^{n-t-1}\d z\\
 & \quad\quad-x^{t}(1-x)^{n-t-1}\Big).
\end{align*}
Since the derivative $\frac{\d}{\d x}x^{t}(1-x)^{n-t-1}$ is given by 
\begin{align*}
 & =tx^{t-1}(1-x)^{n-t-1}-(n-t-1)x^{t}(1-x)^{n-t-2}\\
 & =(t(1-x)-(n-t-1)x)x^{t-1}(1-x)^{n-t-2}\\
 & =(t-(n-1)x)x^{t-1}(1-x)^{n-t-2},
\end{align*}
and the expression $x^{t}(1-x)^{n-t-1}$ can be written as 
\[
\int_{0}^{x}\left(t-(n-1)z\right)z^{t-1}(1-z)^{n-t-2}\d z,
\]
we find that $\hV_{n}'(x)$ can be expressed as 
\begin{align}
\hV_{n}'(x) & =-\frac{1}{2B(t,n-t)}\int_{0}^{x}\bigg(z^{t-1}(1-z)^{n-t-1}\nonumber \\
 & \quad\quad-\left(t-(n-1)z\right)z^{t-1}(1-z)^{n-t-2}\bigg)\d z\nonumber \\
 & =-\frac{1}{2B(t,n-t)}\int_{0}^{x}\bigg((1-z)\nonumber \\
 & \quad\quad-\left(t-(n-1)z\right)z^{t-1}(1-z)^{n-t-2}\bigg)\d z\nonumber \\
 & =-\frac{1}{2B(t,n-t)}\int_{0}^{x}\big(1-t\nonumber \\
 & \quad\quad+(n-2)z\big)z^{t-1}(1-z)^{n-t-2}\d z.\label{eq: P'x-1}
\end{align}
It follows that $\hV_{n}'(x)>0$ for all $0<x<\frac{t-1}{n-2}$ because $z^{t-1}(1-z)^{n-t-2}>0$ and $1-t+(n-2)z<0$ for $0<z<\frac{t-1}{n-2}$. Applying fundamental theorem of calculus to \eqref{eq: P'x-1}, we also find that
\begin{align*}
\hV_{n}''(x) & =\frac{-1}{2B(t,n\!-\!t)}\!\left(1\!-\!t\!+\!(n\!-\!2)x\right)x^{t-1}(1\!-\!x)^{n-t-2},
\end{align*}
and, hence, $\hV_{n}''(x)<0$ for all $\frac{t-1}{n-2}<x<1$. 

Since $\hV_{n}(0)=0$, the bound on $\hV_{n}'(x)$ implies that $\hV_{n}(x)$ is positive for all $0<x<\frac{t-1}{n-2}$. From~(\ref{eq: fixed_point_pot_I}), we see that $\hV_{n}(1)<0$. Thus, $\hV_{n}(x)$ must have a root at some $x^{*}\in[\frac{t-1}{n-2},1]$ and $\hV_{n}'(x^{*})\leq0$ at that root. Since $\hV_{n}''(x)<0$ for all $\frac{t-1}{n-2}\leq x\leq1$, we see that $\hV_{n}'(x)<0$ for all $x>x^{*}$. Thus, the root must be unique. 
\end{IEEEproof}

\section{Proof of Lemma \ref{lem: Vnxn=00003D0}\label{app: pf_Vnx=00003D0}}
\begin{IEEEproof}
Since $\hV_{n}(x)$ is a continuous function on $(0,1]$, we prove that $\hx_{n}^{**}\geq\frac{2t-2}{n}$ for sufficiently large $t$ and $n$ by showing that $\hV_{n}\left(\frac{2t}{n}\right)<0$ and $\hV_{n}\left(\frac{2t-2}{n}\right)>0$ when $t\geq t_{0}$ and $n\geq\min(t+2,n_{0})$.

From (\ref{eq: fixed_point_pot_I}), we observe that
\begin{align}
\hV_{n}\left(\frac{2t}{n}\right) & =\frac{-1}{nB(t,n-t)}\left(\frac{2t}{n}\right)^{t}\left(1-\frac{2t}{n}\right)^{n-t}<0.\label{eq: Vn2tn < 0}
\end{align}
Next, we simplify $V_{n}\left(\frac{2t-2}{n}\right)$ with 
\begin{align*}
\\
 & \hV_{n}\left(\frac{2t-2}{n}\right)=\frac{1}{n}I_{\frac{2t-2}{n}}(t,n-t)\\
 & \qquad-\frac{1}{nB(t,n-t)}\left(\frac{2t-2}{n}\right)^{t}\left(1-\frac{2t-2}{n}\right)^{n-t}\\
 & \ \overset{\mbox{(a)}}{=}\frac{1}{n}I_{\frac{2t-2}{n}}(t,n-t)\\
 & \qquad-\!\frac{t}{n}\!\left(\!1\!-\!\frac{2t\!-\!2}{n}\right)\!\!\binom{n\!-\!1}{t}\!\!\left(\frac{2t\!-\!2}{n}\right)^{\!t\!}\!\!\left(1\!-\!\frac{2t\!-\!2}{n}\right)^{n-t-1}\\
 & \ \geq\frac{1}{n}I_{\frac{2t-2}{n}}(t,n-t)\\
 & \qquad-\frac{t}{n}\binom{n-1}{t}\left(\frac{2t-2}{n}\right)^{t}\left(1-\frac{2t-2}{n}\right)^{n-t-1},
\end{align*}
where (a) follows from the definition of $B(t,n-t)$ in (\ref{eq: Beta}). Assuming $t\geq3$, one can apply the Chernoff bound to the lower tail of the binomial distribution to get 
\begin{align*}
1-I_{\frac{2t-2}{n}}(t,n-t) & \leq\left(2-\frac{2}{t}\right)^{t}e^{-(t-2)}.
\end{align*}
Thus, we know
\begin{align*}
 & \hV_{n}\left(\frac{2t-2}{n}\right)\\
 & \,\geq\!\frac{1}{n}\!\!\left(\!\!1\!-\!\left(\!2\!-\!\frac{2}{t}\!\right)^{\!\!t\!}\!e^{-(t-2)}\!-\!t\binom{\!n\!-\!1\!}{t}\!\!\left(\frac{\!2t\!-\!2\!}{n}\right)^{\!\!t\!}\!\!\left(\!1\!-\!\frac{2t\!-\!2}{n}\!\right)^{\!n-t-1}\!\right)\\
 & \,=\frac{1}{n}\Psi(n;t),
\end{align*}
where $\Psi(n;t)\triangleq1-e^{-(t-2)}\left(2-\frac{2}{t}\right)^{t}-t\binom{n-1}{t}\left(\frac{2t-2}{n}\right)^{t}\left(1-\frac{2t-2}{n}\right)^{n-t-1}$. By the Poisson theorem \cite[pp. 113]{Papoulis-2002} and the fact that $t!\leq t^{t}e^{-t}$ \cite[pp. 30]{RU-2008}, one can show
\begin{align*}
 & \lim_{n\rightarrow\infty}\binom{n-1}{t}\left(\frac{2t-2}{n}\right)^{t}\left(1-\frac{2t-2}{n}\right)^{n-t-1}\\
 & \quad=\frac{(2t-2)^{t}}{t!}e^{-(2t-2)}\\
 & \quad\geq\frac{(2t-2)^{t}}{t^{t}e^{-t}}e^{-(2t-2)}.
\end{align*}
In the limit as $n\rightarrow\infty$ followed by $t\rightarrow\infty$, the function $\Psi(n;t)$ can be lower bounded by 
\begin{align*}
 & \lim_{t\rightarrow\infty}\lim_{n\rightarrow\infty}\Psi(n;t)\\
 & \quad\geq\lim_{t\rightarrow\infty}\left(1-\left(2-\frac{2}{t}\right)^{t}e^{-(t-2)}-t\left(2-\frac{2}{t}\right)^{t}e^{-(t-2)}\right)\\
 & \quad=1-\lim_{t\rightarrow\infty}(1+t)\left(2-\frac{2}{t}\right)^{t}e^{-(t-2)}\\
 & \quad=1.
\end{align*}
Therefore, there exists a $t_{0}\geq3$ and a function $n_{0}(t)\geq t+2$ such that such that $\Psi(n;t)>0$ for all $t\geq t_{0}$ and $n\geq n_{0}(t)$. This implies that $\hV_{n}\left(\frac{2t-2}{n}\right)>0$ for all $t\geq t_{0}$ and $n\geq n_{0}(t)$. Since (\ref{eq: Vn2tn < 0}) holds for any $t>0$ and $n>0$, we conclude that $\hx_{n}^{**}$ exists and $\frac{2t-2}{n}\leq\hat{x}_{n}^{**}\leq\frac{2t}{n}$ for all $t\geq t_{0}$ and $n\geq n_{0}(t)$. 
\end{IEEEproof}

\section{Proof of Lemma \ref{lem: EnQnFinite_n}\label{app: pf_Finite_n_Rec}}
\begin{IEEEproof}
We first rewrite $E\left[\smash{nQ_{n}(X_{n})}\right]$ by
\begin{align*}
 & E\left[nQ_{n}(X_{n})\right]\\
 & \quad=\sum_{i=t+1}^{\lfloor\sqrt{n}\rfloor}\!\binom{n\!-\!1}{i}\!\left(\frac{\lambda}{n\!-\!1}\right)^{i}\!\left(1\!-\!\frac{\lambda}{n\!-\!1}\right)^{n-i-1}nQ_{n}(i)\\
 & \qquad+\!\!\sum_{i=\lfloor\sqrt{n}\rfloor+1}^{n-t-2}\!\!\binom{n\!-\!1}{i}\!\left(\frac{\lambda}{n\!-\!1}\right)^{i}\!\left(1\!-\!\frac{\lambda}{n\!-\!1}\right)^{n-i-1}\!nQ_{n}(i)\\
 & \qquad\qquad+n\!\!\sum_{i=n-t-2}^{n-1}\!\!\binom{n\!-\!1}{i}\!\left(\frac{\lambda}{n\!-\!1}\right)^{i}\!\left(1\!-\!\frac{\lambda}{n\!-\!1}\right)^{n-i-1}.
\end{align*}
From Lemma \ref{lem:LimitOfPn(i)andnQn(i)} and (\ref{eq: finite_n}), we can upper bound $E\left[nQ_{n}\left(X_{n}\right)\right]$ by 
\begin{align}
 & E\left[nQ_{n}(X_{n})\right]\nonumber \\
 & \quad\leq\left(\frac{1}{(t-1)!}+O\left(n^{-0.1}\right)\right)I_{\frac{\lambda}{n-1}}(t+1,n-t-1)\nonumber \\
 & \qquad+C_{1}I_{\frac{\lambda}{n-1}}\left(\lfloor\sqrt{n}\rfloor+1,n-\lfloor\sqrt{n}\rfloor-1\right)\nonumber \\
 & \quad\qquad+nI_{\frac{\lambda}{n-1}}\left(n-t-2,t-2\right),\label{eq: EnQnFinite_n-1}
\end{align}
where $C_{1}$ is a constant. By applying Chernoff bound, the second term of (\ref{eq: EnQnFinite_n-1}) is upper bounded by 
\[
\left(\frac{e\lambda}{\lfloor\sqrt{n}\rfloor}\right)^{^{\lfloor\sqrt{n}\rfloor}}e^{-\lambda}.
\]
Thus, with the upper bound (\ref{eq: BinomTailBound1}) for the last term of (\ref{eq: EnQnFinite_n-1}), we have 
\begin{align*}
 & E\left[nQ_{n}(X_{n})\right]\\
 & \quad\leq\left(\frac{1}{(t-1)!}+O\left(n^{-0.1}\right)\right)I_{\frac{\lambda}{n-1}}(t+1,n-t-1)\\
 & \qquad+C_{1}\left(\frac{e\lambda}{\lfloor\sqrt{n}\rfloor}\right)^{\lfloor\sqrt{n}\rfloor}e^{-\lambda}+n\left(\frac{e\lambda}{n-1-t}\right)^{n-t-1}e^{-\lambda}
\end{align*}
Next, we observe that, for any $\lambda>0$, 
\begin{align*}
 & \lim_{n\rightarrow\infty}\frac{\left(\frac{e\lambda}{\lfloor\sqrt{n}\rfloor}\right)^{\lfloor\sqrt{n}\rfloor}}{I_{\frac{\lambda}{n-1}}(t+1,n-t-1)}\\
 & \qquad\leq\lim_{n\rightarrow\infty}\frac{\left(\frac{e\lambda}{\lfloor\sqrt{n}\rfloor}\right)^{\lfloor\sqrt{n}\rfloor}}{\left(\frac{\lambda}{n-1}\right)^{t+1}\left(1-\frac{\lambda}{n-1}\right)^{n-t-2}}=0,
\end{align*}
and 
\begin{align*}
 & \lim_{n\rightarrow\infty}\lim_{\lambda\rightarrow0}\frac{\left(\frac{e\lambda}{\lfloor\sqrt{n}\rfloor}\right)^{\lfloor\sqrt{n}\rfloor}}{I_{\frac{\lambda}{n-1}}(t+1,n-t-1)}\\
 & \quad\leq\lim_{n\rightarrow\infty}\lim_{\lambda\rightarrow0}\frac{\left(\frac{e\lambda}{\lfloor\sqrt{n}\rfloor}\right)^{\lfloor\sqrt{n}\rfloor}}{\left(\frac{\lambda}{n-1}\right)^{t+1}\left(1-\frac{\lambda}{n-1}\right)^{n-t-2}}\\
 & \quad=\lim_{n\rightarrow\infty}\lim_{\lambda\rightarrow0}\frac{(n-1)^{t+1}\left(\frac{e}{\lfloor\sqrt{n}\rfloor}\right)^{\lfloor\sqrt{n}\rfloor}\lambda^{\lfloor\sqrt{n}\rfloor-t-1}}{\left(1-\frac{\lambda}{n-1}\right)^{n-t-2}}\\
 & \quad=\lim_{\lambda\rightarrow0}\lim_{n\rightarrow\infty}\frac{(n-1)^{t+1}\left(\frac{e}{\lfloor\sqrt{n}\rfloor}\right)^{\lfloor\sqrt{n}\rfloor}\lambda^{\lfloor\sqrt{n}\rfloor-t-1}}{\left(1-\frac{\lambda}{n-1}\right)^{n-t-2}}\\
 & \quad=0.
\end{align*}
Thus, we have $\left(\left(\frac{e\lambda}{\lfloor\sqrt{n}\rfloor}\right)^{\lfloor\sqrt{n}\rfloor}+n\left(\frac{e\lambda}{n-1-t}\right)^{n-t-1}\right)I_{\frac{\lambda}{n-1}}^{-1}(t+1,n-t-1)=O(n^{-1})$ and 
\begin{align*}
 & E\left[nQ_{n}(X_{n})\right]\\
 & \quad\leq\left(\frac{1}{(t-1)!}+O\left(n^{-0.1}\right)\right)I_{\frac{\lambda}{n-1}}(t+1,n-t-1).
\end{align*}
This concludes the proof of the lemma.
\end{IEEEproof}


\end{document}